\documentclass[runningheads]{llncs}

\newif\ifanonymous
\anonymousfalse  

\usepackage[T1]{fontenc}
\usepackage{lmodern}
\usepackage{graphicx}
\usepackage{amssymb}
\usepackage{booktabs}
\usepackage{hyperref}
\usepackage{listings}
\usepackage{latexsym, amsfonts, yfonts, bbm}
\usepackage{stmaryrd}
\usepackage{enumerate}
\usepackage{mathrsfs}
\usepackage{amsmath}
\makeatletter
\let\c@lemma\c@theorem

\@ifundefined{c@definition}{}{\let\c@definition\c@theorem }
\@ifundefined{c@proposition}{}{\let\c@proposition\c@theorem }
\@ifundefined{c@corollary}{}{\let\c@corollary\c@theorem }
\@ifundefined{c@remark}{}{\let\c@remark\c@theorem }
\@ifundefined{c@example}{}{\let\c@example\c@theorem }
\makeatother

\hypersetup{
  colorlinks=true,
  linkcolor=blue,
  citecolor=blue,
  urlcolor=blue,
  pdfauthor={},
  pdftitle={}
}

\def\ok#1{\mbox{\raisebox{0ex}[1ex][1ex]{$#1$}}}

\def\re{c.e.}
\def\Imp{\tt Imp}

\def\Aexp{\mathit{a}}
\def\Bexp{\mathit{b}}

\usepackage{thcomp}
\newcommand{\myqed}{\textsc{Q.E.D.}}
\def\wpre{\wp^{\text{ce}}}
\def\wprec{\wp^{\text{dec}}}
\def\grasseb#1{{\llparenthesis\hspace*{0.2ex} #1 \hspace*{0.2ex}\rrparenthesis}}
\def\grasse#1{{\llbracket #1 \rrbracket}}

\newcommand{\li}{\ar@{-}} 
\newcommand{\lp}{\ar@{.}} 
\newcommand{\fp}{\ar@{.>}} 

\def\bbbc{{\mathbb{C}}}

\def\dovetl{{\tt dovetail}}

\newcommand{\Interp}{\mathsf{U}}

\newcommand{\ccL}{\mathscr{L}}
\newcommand{\ccD}{\mathscr{D}}

\newcommand{\scl}{\mathbb{S}}

\newcommand{\Domain}{\mathbb{D}}

\newcommand{\NF}{{\mathcal{T}}}
\newcommand{\pfunction}[3]{\ok{\varphi^{#2\, #3}_{#1}}}

\newcommand{\auco}[2]{\ok{\rho^{#1}_{#2}}}

\newcommand{\compl}[1]{\ok{\overline{#1}}}

\newcommand{\uco}{\ensuremath{\mathrm{uco}}}

\newcommand{\Var}{\ensuremath{\mathit{Var}}}

\newcommand{\ra}{\rightarrow}

\newcommand{\Ra}{\Rightarrow}

\newcommand{\Lra}{\Leftrightarrow}

\def\defi{\triangleq}
\newcommand{\range}{\textit{rng}}
\newcommand{\dom}{\textit{dom}}

\DeclareMathOperator{\Abs}{Abs}

\newcommand{\integer}{\mathbb{Z}}
\newcommand{\ud}{\ok{\:\triangleq\:}}

\def\ok#1{\mbox{\raisebox{0ex}[1ex][1ex]{$#1$}}}

\def\predC{\mathscr{C}}
\def\lfp{{\sf lfp\/}}
\def\gfp{{\sf gfp\/}}

\def\comp{\hbox{\;\footnotesize$\circ$\normalsize}}

\def\defemb#1#2{\expandafter\def\csname #1\endcsname
                              {\relax\ifmmode #2\else\hbox{$#2$}\fi}}
 
\def\2c-math#1#2{{\par\medskip\noindent ${#1}$
                      \par\smallskip
                        \noindent\hspace*{\fill} ${#2}$}
                           \\[10pt]}
\def\cAA{\ensuremath{A}}
                           
\defemb{aF}{{\cal F}^\cAA}
\defemb{sF}{{\cal F}^{\star}}

\defemb{cB}{{\cal B}}
\defemb{cC}{{\cal C}}
\defemb{cD}{{\cal D}}
\defemb{cE}{{\cal E}}
\defemb{cF}{{\cal F}}
\defemb{cG}{{\cal G}}
\defemb{cH}{{\cal H}}
\defemb{cI}{{\cal I}}
\defemb{cJ}{{\cal J}}
\defemb{cK}{{\cal K}}
\defemb{cL}{{\cal L}}
\defemb{cM}{{\cal M}}
\defemb{cN}{{\cal N}}
\defemb{cO}{{\cal O}}
\defemb{cP}{{\cal P}}
\defemb{cQ}{{\cal Q}}
\defemb{cR}{{\cal R}}
\defemb{cS}{{\cal S}}
\defemb{cT}{{\cal T}}
\defemb{cU}{{\cal U}}
\defemb{cV}{{\cal V}}
\defemb{cW}{{\cal W}}
\defemb{cY}{{\cal Y}}
\defemb{cX}{{\cal X}}
\defemb{cZ}{{\cal Z}}
\defemb{cHys}{{\cH_{\cY}^{\cS}}}
\defemb{cHy}{{\cH_{\cY}}}
\defemb{ld}{\dashv}
\defemb{rd}{\vdash}
\defemb{cGH}{{\cal GH}}
\defemb{cGWP}{{\cal GWP}}
\defemb{tL}{{\tt L}}
\defemb{tH}{{\tt H}}
\defemb{tT}{{\tt T}}
\defemb{tS}{{\tt S}}

\newcommand{\sset}[2]{\{#1  ~|~ #2\}}

\newcommand{\nat}{\omega}

\def\tuple#1{\langle #1 \rangle}

\makeatletter
\newcommand{\pto}{}
\newcommand{\pgets}{}
\DeclareRobustCommand{\pto}{\mathrel{\mathpalette\p@to@gets\to}}
\DeclareRobustCommand{\pgets}{\mathrel{\mathpalette\p@to@gets\gets}}
\newcommand{\p@to@gets}[2]{%
  \ooalign{\hidewidth$\m@th#1\mapstochar\mkern5mu$\hidewidth\cr$\m@th#1\to$\cr}%
}
\makeatother

\newcommand{\defined}{{\downarrow}}

\makeatletter
\newcommand{\TODO}{\@ifnextchar[{\@TODO}{\@@TODO}}
\long\def\@TODO[#1]#2{\bgroup\color{red}\ifmmode\else\itshape\fi TODO[#1]: #2\egroup}
\long\def\@@TODO#1{\bgroup\color{red}\ifmmode\else\itshape\fi TODO: #1\egroup}
\makeatother

\begin{document}

\title{Elusive but Coverable: The Recursion-Theoretic Structure of Complete Abstract Interpretations
}
\titlerunning{The Recursion-Theoretic Structure of Complete Abstract Interpretations}

\ifanonymous
  \author{Anonymous Authors}
  \authorrunning{Anonymous Authors}
  \institute{Anonymous Institution}
\else
  \author{Nicklas Carpenter \and Roberto Giacobazzi}
  \authorrunning{}
  \institute{
    University of Arizona\\
    \email{$\{$npcarpenter\,|\,giacobazzi$\}@$arizona.edu}
  }
\fi 

\maketitle

\begin{abstract}
We study local completeness and incompleteness of abstract interpretations from a recursion-theoretic perspective. Local completeness weakens global completeness and captures the absence of precision loss for a specific precondition: abstract computation yields exactly what is obtained by abstracting the corresponding concrete computation. This enables compositional reasoning and rules out false positives in verification.
We characterize the distinction between static and dynamic program analysis in terms of uniformly decidable operations and observe that the latter is uniformly decidable only for trivial abstractions. We then prove that the class of programs inducing a predicate transformer that is locally complete for a given non-trivial abstract domain is elusive in a precise recursion-theoretic sense: it is a productive set, hence not computably enumerable, and, under mild hypotheses, the same holds for its complement. In particular, the first class lies in $\Pi^0_2$
and the second in $\Sigma^0_2$.
Unlike the usual examples of $\Pi^0_2$ properties, we show that the classes of locally complete programs admit decidable coverings. This makes it possible to construct, via program transformation, an effective enumeration of a representative subset of programs that entirely covers this class --- capturing from the outside a class that eludes enumeration from within.
\keywords{Abstract interpretation \and completeness \and program analysis \and program verification \and computability theory}
\end{abstract}

\section{Introduction}

In any theory, the study of limit cases is a powerful means of exposing its conceptual structure. They mark the boundary between possibility and impossibility, distinguish essential assumptions from accidental ones, and provide idealized reference points for understanding more practical instances. Complete abstract interpretations play this role within abstract interpretation. 

\smallskip

Abstract interpretation is a general theory for specifying approximate semantics of programming languages \cite{CC77}, including program analysis, program logics, and program transformations as special cases. 
This is achieved by designing an approximate (abstract) interpreter from the semantics of a given computational system (e.g., a programming language) and an abstraction of its state space. The abstraction, specified by the so called {\em abstract domain\/}, later denoted $\cAA$, plays a key role in the abstract interpretation construction. It specifies precisely what information has to be retained and what can be abstracted away, with the specific goal that interpreting in the simplified abstract domain yet keeps as much information as possible to achieve our goals, e.g., proving program correctness. 

\smallskip

By construction, abstract interpretation is sound \cite{CC77}, meaning that if the approximate abstract interpretation of a program is correct with respect to a specification then also its concrete semantics satisfies the specification.
The converse is rare and corresponds to the limit, and ideal, situation where no loss of precision is cumulated in the abstract interpretation with respect to what can be obtained by abstracting the concrete computation. When this happens we have {\em completeness} \cite{CC79,GRS00}. 
Completeness clarifies which aspects of a computation can be abstracted away while preserving the ability to reason precisely about its properties.
In program verification, program analysis and program logics completeness guarantees the absence of false positives, therefore making it possible to have a perfect matching between concrete and approximate correctness proofs. It is about three decades that completeness has been studied as a field by its own \cite{GR97b}. 
It has been extensively studied in the context of abstraction refinement for program analysis \cite{GRS00,GQ01,BruniGGR22}, information flow security by (abstract) non-interference \cite{GMadj10}, code protection \cite{G08}, semantic properties of abstract interpreters, such as extensionality and referential transparency \cite{BruniGGGP20}, in the combination of correctness and incorrectness logics of programs \cite{BruniGGR21,BruniGGR23}, in the coordination of inductive and co-inductive-{\em up-to\/} techniques for proving properties of greatest fixed point computations \cite{BGGP18} and, recently, in the context of the design of {\em best correct approximations} (bca) of predicate transformers \cite{GiacobazziR25}.

\subsection{The problem}
On one side standard computability is a theory of sets and functions indexed on programs enjoying all beautiful properties such as extensionality and referential transparency \cite{odifreddi}. On the other side 
it has been proved that abstract interpretation is inherently intensional in nature \cite{BruniGGGP20}, i.e., its quality strictly depends upon the way the code is written. This is indeed a direct consequence of the lack of compositionality \cite{GiacobazziR25} which is guaranteed only when we can suppress intermediate abstractions in the composition of abstract functions, as it happens when the outer function is complete {\em locally\/} to the output of the inner one. It is precisely completeness in this localized version that plays a central role in reconciling computability theory with abstract interpretation. 

\smallskip

Local completeness is a weaker form of completeness introduced in \cite{BruniGGR21} and further developed in \cite{BruniGGR23,BruniGGR22} in the contexts of LCL (Local Completeness Logic), combining correctness and incorrectness logic, and AIR (Abstract Interpretation Refinement), generalizing CEGAR-style refinement to arbitrary abstract domains.
Unlike global completeness, traditionally studied in the literature (e.g., see \cite{GRS00}), which requires completeness with respect to all possible inputs, local completeness requires completeness only for a fixed input property. When it holds along a computation trace in the intermediate junctions of post/pre conditions, it guarantees that no loss of precision is propagated along that trace. This localized perspective has proved particularly useful for analyzing the precision of abstract interpreters defined inductively over the syntax of programs.
Contrary to global completeness studied in \cite{GiacobazziLR15}, local completeness does not have a recursion-theoretic characterization. Our contribution, which is primarily theoretical in nature, wants to fill this gap.

\subsection{Main contribution}
%


We first analyze decidable abstract domains as uniform transformations. 
For decidable abstract domains, whose elements represent decidable program properties as commonly used in program analysis and verification, we show that uniformity precisely separates dynamic from static program analysis. Here, uniformity is the ability to compute effectively the abstract property associated with a concrete set from any program enumerating that set, independently of how the set is generated or used. This establishes the boundary between merely abstracting the result of an interpreter and performing abstract interpretation.

\smallskip

In this context we study the complexity in the arithmetical hierarchy of the classes of all programs admitting an abstract semantics which is locally complete for a given input property $I$ in a fixed abstract domain $\cAA$. We prove that when $\cAA$ is not trivial, i.e., it is not the identical abstraction or the abstraction which collapses all properties into a single {\em don’t know} value,
this class, later denoted $\bbbc(\cAA,I)$, is bi-productive, i.e., it is productive as well as its complement, the latter under mild hypothesis. Only when $I=\varnothing$ the complement class $\compl{\bbbc(\cAA,I)}$, called the {\em incompleteness class}, is computably enumerable. In particular we prove that for continuous abstract domains and input property $I\neq\varnothing$ the classes $\bbbc(\cAA,I)$ and $\compl{\bbbc(\cAA,I)}$ are respectively in $\Pi^0_2$ and $\Sigma^0_2$ in the arithmetic hierarchy. 

\smallskip

The above results show the intrinsic difficulties associated with completeness and incompleteness persist even in their weaker---local, formulation. Since productive sets are by definition not computably enumerable, neither local completeness nor local incompleteness admits an effective enumeration that fulfills respectively $\bbbc(\cAA,I)$ and $\compl{\bbbc(\cAA,I)}$. Actually these results show that these two classes are recursively inseparable, namely no decidable set $C$ that contains all elements of $\bbbc(\cAA,I)$ while excluding all elements of $\compl{\bbbc(\cAA,I)}$; that is, no computable yes/no procedure can perfectly separate the positive cases in $\bbbc(\cAA,I)$ from the negative cases in $\compl{\bbbc(\cAA,I)}$.
While productivity rules out any effective enumeration of these classes as sets of programs, we prove that $\bbbc(\cAA,I)$ admits a computably enumerable (and hence also a decidable) covering. The idea here is that we do not try to enumerate exactly all and only the programs in $\bbbc(\cAA,I)$. Instead, we provide a fix-point characterization in terms of effective program transformations\footnote{We use the term {\em program transformation\/} in a non-standard sense: we allow a transformation to modify not only the syntactic form of a program, but also its semantics.} of a representative family of programs whose semantic behavior covers the entire class. The covering is allowed to miss many programs, as long as it still contains enough representatives to capture all relevant behaviors in $\bbbc(\cAA,I)$. The proof follows the structure of similar proofs in abstract computational complexity (e.g., see \cite{Blum67}): here to be cancelled by an effective program transformation are the cases of incompleteness.

\subsection{Related works}

The most related paper is \cite{GiacobazziLR15} where the authors firstly analyzed the problem of completeness from the perspective of the programs for which an abstract domain is complete instead of the abstract domains that induce completeness for a given program, as previously extensively done in the literature. This naturally leads to the notion of completeness class. We extend their approach to the weaker case of local completeness and we prove that local and global completeness share the same complexity in the arithmetic hierarchy. Moreover we characterize a {\re} (computably enumerable) covering of the class of local complete programs for a non trivial abstract domain and prove the complexity of transformations mapping arbitrary programs to the complement class of locally incomplete programs. The covering provides the very first characterization of completeness in fix-point form, while the transformation making programs locally incomplete  provides an effective lower bound to code obfuscating transformations.

In \cite{DMonniaux23} the author presents an important practice-oriented (and personal) account of completeness in static analysis by abstract interpretation, distinguishing the completeness of an abstract domain from the completeness of the method used to compute invariants. The author affirms that completeness matters because incompleteness methods may be hard to control and unpredictable: for example, widening can fail to find invariants that exist in the chosen abstract domain, and small or seemingly irrelevant program changes may alter the analysis outcome. The author concludes that complete methods, when available, provide robustness, unique well-defined results, better testability, and clearer complexity questions. 
With respect to \cite{DMonniaux23} we consider only imprecise transfer functions and the lack of compositionality due to combinations of abstractions as possible sources of incompleteness. Interestingly the productivity of $\bbbc(\cAA,I)$ is consistent with the undecidability of the existence of suitable invariants. This connection opens interesting future developments on the expressivity of abstract domains to represent program invariants.

In \cite{GiacobazziM16} the authors studied semantic transformations that minimally modify a semantic function to make it globally complete for a fixed abstract domain. These results can be seen as a model-based attempt to cover the class of all predicate transformers that are globally complete. The limit of that approach is in its pure model-based nature: it transforms the semantics to achieve completeness, not the code that computes that semantics. 
We instead attack the problem from a recursion-theoretic perspective, namely by means of program transformations. By characterizing effective program transformations that remove local incompleteness we are able to precisely identify a computably enumerable covering of $\bbbc(\cAA,I)$, hence achieving local completeness by code repair.

\section{Preliminaries}
In this section we introduce some mathematical notation and background in recursion theory and abstract interpretation. For the first the reader can refer to \cite{odifreddi,soare,rogers}. For the latter see \cite{CC77,CC79,Cousotbook}.

\subsection{Sets and functions}
Throughout this work, $\nat$ will denote the set of all natural numbers and $\aleph_0$ its cardinality. 
Given two sets $S$ and $T$, $\wp(S)$ denotes the powerset of $S$, $S \smallsetminus T$ denotes the set-difference between $S$ and $T$, 
$\compl{S}$ denotes the complement of $S$ with respect to the relevant universe of 
discourse, and $S \subsetneq T$ denotes strict inclusion. We use exponents to
represent the repeated Cartesian product of a set with itself, so
\(S \times S = S^2 \) and \(S^{n+1} = S^n \times S\) for all \(n \geq 2\).

A function \(f\) is introduced $f:S\ra T$ if \(f\) is a total function from \(S\)
to \(T\) or $f: S\pto T$ if it is a partial function. For any \(x \in S\), we write
$f(x){\downarrow}$ to signify that \(f\) is defined for the particular \(x\) and
$f(x){\uparrow}$ to signify that $f(x)$ is not defined. The domain and range of $f$ are, respectively, the sets
\[\dom(f) \ud \sset{x \in S }{ f(x){\downarrow}}~~\mbox{and}~~\range(f)\ud\sset{f(x)}{x \in S\cap \dom(f)}.\] 
Given a set $X\subseteq S$, the image of \(f\) on \(X\) is
\[\ok f(X) \ud \{f(x)~|~x\in X\cap \dom(f)\}.\]
Two partial functions $f,g: S\pto T$ are extensionally equivalent, denoted $f\cong g$, if $\dom(f)=\dom(g)$ and \(f(x) = g(x)\) for all \(x \in \dom(f) = \dom(g)\).

Occasionally, we define functions using $\lambda$-notation $\lambda x \st f(x)$; 
often, this notation is used when we want to emphasize the arguments of the function.
For any two functions $f : S \pto T$ and $g : T \pto U$, $g \circ f : S \pto U$, we write
\((\fcomp{f}{g})(x)\) to denote the composition of \(f\) and \(g\) where
\((\fcomp{f}{g})(x)\defined\) when both $g(x)\defined$ and $f(g(x))\defined$; otherwise,
$(f \circ g)(x){\uparrow}$. 

\subsection{Basic Order Theory}

A set $L$ endowed with a partial order relation $\leq$ is called a poset and 
is denoted $\tuple{L,\leq}$. A \emph{chain} of the poset $\tuple{L,\leq}$ is a 
sequence of elements in $L$ totally-ordered under the relation \(\leq\).
A poset $\tuple{L,\leq}$ is a \emph{lattice} whenever binary least upper bounds (lubs)
$x \vee y$ and greatest lower bounds (glbs) $x \wedge y$ exist for all  $x,y \in L$. 
A lattice $\tuple{L,\leq}$ is \emph{complete} if every subset $X \subseteq L$ (including
the empty set) has a unique lub $\bigvee X$ and glb $\bigwedge X$. We write $\tuple{L, 
\leq, \vee, \wedge, \top, \bot}$ to denote the complete lattice $L$ with partial order
$\leq$, lub $\vee$, glb $\wedge$, greatest element (top) $\top$, and least element
(bottom) $\bot$. A poset $\tuple{L,\leq}$ is an $\omega$-cpo if $L$ contains the lubs of all $\omega$-chains---chains that can be indexed by the natural numbers.

If $\tuple{L,\leq}$ is a poset and $f,g:S \ra L$, then \(f\) and \(g\) can be
ordered with respect to the poset; we write $f \sqsubseteq g$ if
$f(x) \leq g(x)$ for all $x \in S$. If $L$ is a (complete) lattice 
then $\tuple{S\ra L,\sqsubseteq}$ is a (complete) lattice; the operators
$\sqcup$ and $\sqcap$ are used to denote, respectively, the lub or glb of \(L\)
or any of its chains or subsets. A function $f : L_1 \ra L_2$ between complete lattices is additive (co-additive) 
if for all $Y \subseteq L_1$, $f(\vee_{L_1} Y ) = \vee_{L_2} f(Y)$ ($f(\wedge_{L_1} Y ) = \wedge_{L_2} f(Y ))$. Also, $f$ is continuous (co-continuous) when $f$ preserves lubs (glbs) of chains in $L_1$. Recall that any monotone function $f:L \ra L$ on a complete lattice $L$ 
always has least and greatest fix-points, denoted $\lfp(f)$ and $\gfp(f)$
respectively. If $f:L \ra L$ is continuous then  
$\lfp(f)=\bigvee_{n\in\nat}f^n(\bot)$ where \(f^n\) is defined inductively so that 
$f^0(x)\ud x$ and $f^{n+1}(x)\ud f(f^n(x))$ for all \(n \in \nat\). 

\subsection{Basic recursion theory}\label{comp-sec}

\subsubsection{Computable functions.}
Partial computable functions are (possibly partially-defined) functions that can
be computed by a Turing Machine. We assume a surjective enumeration
$\ccL$, also called a {\em programming system\/}, of partial computable functions
of $n$ arguments and $m$ outputs for any $n,m\in\nat$. For any $e\in\nat$, the
function $\pfunction{e}{n}{m}$ denotes a partial computable function of index $e$ taking $n$ arguments and producing $m$ outputs. Here, $e$ represents the code of a unique Turing machine that computes $\pfunction{e}{n}{m}$.
Following standard convention, we assume the partial functions $\pfunction{e}{n}{m}$ are defined over the infinite denumerable domains $\Domain^n$ and $\Domain^m$
(i.e., for all $n\in\omega:\,|\Domain^n|=\aleph_0$) where \(\Domain\) includes
a representation of all natural numbers and the programming system $\ccL$ (i.e.,
$\nat \subseteq \Domain$ and $\ccL\subseteq\Domain$). Under the convention, the
partial computable functions are $\pfunction{e}{n}{m}: \Domain^n \pto \Domain^m$,
for all $e\in\ccL$ and all $n,m\in\nat$. In practice, we omit $n$ and/or $m$ when
the number of arguments and outputs of function is irrelevant or clear from context. 

A programming system $\ccL$ is considered {\em acceptable\/} if 
it is {\em universal\/}; that is, if it includes a partial computable function $\Interp:\ccL\times\Domain\pto\Domain$ such that if $\{\psi_i\}_{i\in\omega}$ is an enumeration of all partial computable functions then $\Interp(i,x)=\psi_{i}(x)$, and 
there exists a total computable function $t:\omega\ra\ccL$ such that $\varphi_{t(i)}=\psi_{i}$.
Because acceptable enumerations are isomorphic by Rogers' isomorphism, we can identify the index $e\in \ccL$ as the corresponding program in the given programming system $\ccL$, which is usually a programming language. Hence 
$\pfunction{e}{}{}$ is the function (semantics) computed by the program $e\in\ccL$.
A set of programs $S\subseteq\ccL$ is an {\em extensional property\/} of $\ccL$ if 
$e\in S~\wedge~\pfunction{e}{}{}\cong\pfunction{i}{}{}~\Ra~i\in S$. Two sets $A,B\subseteq\ccL$ are extensionally equivalent if 
$\sset{\pfunction{e}{}{}}{e\in A}=\sset{\pfunction{e}{}{}}{e\in B}$. $A$ is a {\em covering\/} (a.k.a, representation) of 
$B$ when $A$ and $B$ are extensionally equivalent and $A\subsetneq B$.

\subsubsection{Computability.}
A set $S\subseteq \Domain$ is computably enumerable (\re) if there exists a partial computable function $\pfunction{e}{}{}$ such that $S=\dom(\pfunction{e}{}{})$. It is known that $S$ is {\re} iff $S = \pfunction{e}{}{}(\Domain)\ud\range(\pfunction{e}{}{})$ for some total computable function $\pfunction{e}{}{}$.
For a partial computable function $\pfunction{e}{}{}$ we denote $W_e = \dom(\pfunction{e}{}{})$. $S\subseteq \Domain$ is {\em decidable} if and only if $S$ and $\compl{S}$ are \re.
The set of all {\re} sets in $\Domain$ is $\wpre(\Domain)\ud \sset{W_e \in \wp(\Domain)}{e\in\ccL}$. 
It is known \cite[Union Theorem~1.9, Chapter~II]{soare} 
that $\tuple{\wpre(\Domain),\subseteq}$ is 
a distributive lattice with $\varnothing$ and $\Domain$ as, respectively, bottom and top elements and that
the set of decidable sets $\wprec(\Domain) \ud \sset{S\in \wp(\Domain)}{
S, \compl{S}\in \wpre(\Domain)}$
is a Boolean algebra $\ok{\tuple{\wprec(\Domain),\subseteq}}$. 
A sequence $\ok{\{V_e\}_{e\in\nat}\subseteq\wpre(\Domain)}$ is {\em uniformly {\re}} if there exists a total computable function $f:\nat\ra\nat$ such that $V_e=W_{f(e)}$. Such sequence $\{W_{f(e)}\}_{e\in\nat}$ is uniformly decidable if the set $\sset{(x,e)}{x\in W_{f(e)}}$ is decidable.

\smallskip

It is known that for any $n\in\nat$ there exists a primitive computable predicate $\NF$ of $n+2$ arguments such that 
for each $e\in\ccL$ and $x\in\Domain$: $\pfunction{e}{n}{}(x){\downarrow} \Leftrightarrow \exists t.\NF(e,x,t)$ \cite{kleene38}. Here $\NF(e,x,t)$ may represent the primitive computable predicate which holds true if the program $e\in\ccL$ has terminated its computation in $t$-steps in some model of computation for the programming system $\ccL$, yet returning the value $\pfunction{e}{n}{}(x)$. 
Given $e\in\ccL$ ant $t\in\omega$ we denote by $\dovetl(e,t)\subsetneq\Domain$ the 
set of all inputs $x\in \Domain$ determined in $t$-steps of the dovetail procedure such that $\pfunction{e}{}{}(x){\downarrow}$. It is known that if $e\in\ccL$ then $W_e=\cup_{t\in\omega}\dovetl(e,t)$.

\smallskip

Not all elements in $\wp(\Domain)$ are \re: 
both $\wpre(\Domain)$ and $\wprec(\Domain)$ are denumerable and $\wprec(\Domain)\subsetneq\wpre(\Domain)\subsetneq\wp(\Domain)$.
A set $S\in\wp(\Domain)$ is {\em creative\/} if it is {\re} and its complement $\compl{S}=\Domain\smallsetminus S$ is {\em productive\/} \cite{Post44,Myhill55}, i.e., there exists a total computable function $\pfunction{e}{}{}:\ccL\ra\Domain$ such that
$\forall x\in\ccL. \, W_x\subseteq \compl{S}~\Rightarrow~\pfunction{e}{}{}(x)\in\compl{S}\smallsetminus W_x$.
It is clear that all productive sets are not {\re}. It is also known that, while the complement of a creative set is always productive, the complement of a productive set may be productive, e.g., $\sset{e\in \ccL}{\mbox{$W_e\in\wprec(\Domain)$}}$ \cite{Dekker55}. It is also known that creative sets are {\em complete\/} in $\wpre(\Domain)$, i.e., $S\in \wp(\Domain)$ is creative iff $S\in\wpre(\Domain)$ and
for any $A\in\wpre(\Domain)$, $A\preceq_f S$. Here for $A\in \wp(\Domain)$ and $S\in \wp(\Domain)$, $A\preceq_f S$ denotes {\em many-to-one reducibility}, i.e., the existence of a total computable function $f:\Domain \ra \Domain$ such that for all $x\in \Domain$: $x\in A~\Leftrightarrow~f(x)\in S$. It is known that if $S$ is productive and $S\preceq_f X$ then also $X$ is productive. It is known that $\ok{K\ud\sset{e}{\pfunction{e}{}{}(e){\downarrow}}\in\wpre(\Domain)}$ is creative and therefore
$\ok{\compl{K}}$ is productive.

\smallskip

The Kleene arithmetical hierarchy is particularly important to compare properties of programs. 
We denote by $\Sigma^0_0=\Pi^0_0=\Delta_0$ the set of all decidable sets, i.e., $\Sigma^0_0=\Pi^0_0=\Delta_0 = \wprec(\Domain)$. 
For $n\geq 1$ we define when a set $A\in \wp(\Domain)$ is arithmetical as follows:
\begin{itemize}
\item $A\in\Sigma^0_n$ if there exists a decidable predicate $R(x,y_1,\ldots,y_n)\subseteq \Domain^{n+1}$:
\begin{align*}
x\in A~\Lra~\exists y_1,\forall y_2\ldots Q y_n.\;& R(x,y_1,\ldots,y_n)\\
&\mbox{where $Q=\exists$ if $n\in 2\nat+1$, and $Q=\forall$ if $n\in 2\nat$.}
\end{align*}
\item $A\in\Pi^0_n$ if there exists a decidable predicate $R(x,y_1,\ldots,y_n)\subseteq \Domain^{n+1}$:
\begin{align*}
x\in A~\Lra~\forall y_1,\exists y_2\ldots Q y_n.\;& R(x,y_1,\ldots,y_n)~\\
&\mbox{where $Q=\forall$ if $n\in 2\nat+1$, and $Q=\exists$ if $n\in 2\nat$.}
\end{align*}
\item $A\in\Delta_n$ if $A\in\Sigma^0_n\cap\Pi^0_n$.
\end{itemize}
We know $A\in\Sigma^0_n~\Lra~\compl{A}\in\Pi^0_n$, 
$A\in\Sigma^0_n\cup\Pi^0_n\Rightarrow\forall m>n.\; A\in\Delta_{m}=\Sigma^0_m\cap\Pi^0_m$, 
$B\preceq_f A~\wedge~ A\in\Sigma^0_n\Rightarrow B\in\Sigma^0_n$, and
 $\ok{R\in\Sigma^0_{n>0}}~ \Ra~ \sset{x}{\exists y.\;R(x,y)}\in\Sigma^0_n$ \cite{rogers}.

\subsubsection{Program semantics.}
Although in the following $\ccL$ can be any acceptable programming system, we sometimes consider $\ccL=\Imp$, where $\Imp$ is a simple {\tt while}-programming language with arithmetic ${\tt Exp}$ and Boolean ${\tt BExp}$ expressions (see \cite{winskel}), whose syntax is as follows: 
  \begin{align*}
   {\tt Exp} \ni a  ::= &\; v\in\Domain\mid x\in \Var\mid  f(a)\mid W \\
   {\tt BExp} \ni b ::= &\; {\tt true}\mid {\tt false} \mid a =a \mid a > a\mid b\wedge b\mid\neg b\mid R\\
    {\Imp}\ni p ::= &\;{\tt skip}\mid x:=\Aexp\mid p;p\mid 
    {\tt if} ~b~ {\tt then} ~p~ {\tt else} ~p \mid 
    {\tt while} ~b~ \{ ~p~\}
   \end{align*}
where $f$ ranges over partial computable functions, $R$ ranges over decidable predicates and
$W$ ranges over arbitrary decidable sets of elements in $\Domain$. 
In this case $\Domain$ contains basic values (e.g., integers in $\integer$, booleans in $\{{\tt true, false}\}$, finite sets, etc.). Let $\scl$ denotes the set of all stores $\sigma:\Var\ra\Domain$, i.e., total computable functions that assign values in $\Domain$ to finite vectors of variables in $\Var$. When $S\subseteq\scl$ is a set of stores then a predicate transformer semantics of a program $p\in\Imp$ can be inductively defined as usual as follows $\grasse{p}: \wp(\scl)\ra \wp(\scl)$ where:
\begin{align*}
\grasse{{\tt skip}}S &\defi S\\
\grasse{x:=\Aexp}S &\defi \{\ \sigma[x\mapsto \grasse{\Aexp}\sigma] \mid \sigma \in S\ \}\\
\grasse{p_1 ; p_2 } S &\defi \grasse{p_2}(\grasse{p_1}S)\\
\grasse{{\tt if}\ \Bexp\ {\tt then}\ p_1\ {\tt else}\ p_2}S &\defi  \grasse{p_1}(\grasse{\Bexp}S) \cup \grasse{p_2}(\grasse{\neg\Bexp}S)\\
\grasse{{\tt while}\ \Bexp\  \{\ p\ \} }S &\defi \grasse{\neg \Bexp}\big(\lfp (\lambda T.\: S \cup \grasse{p}(\grasse{\Bexp}T))\big).
\end{align*}
In this case
$\grasse{b}S=\sset{\sigma\in S}{\sigma\models b}$, where $\sigma\models b$ holds when the store $\sigma$ viewed as a substitution satisfies $b$, and $\grasse{a}S=\sset{\grasseb{a}\sigma}{\sigma\in S}$, where $\grasseb{a}\sigma\in\Domain$ is the value that the expression $a$ can take with store $\sigma$. We assume the total computable functions ${\tt input}:\Var\times\Domain\ra\scl$ and ${\tt out}:\Var\times\scl\ra\Domain$ such that 
${\tt input}(x,v)$ builds the input store mapping $x$ to $v$ and
${\tt out}(x,\sigma) = \sigma(x)$, where $x$ is some output vector of variables. We assume that the output vector of variables belongs to the set of free variables of the program $e$, devoted $FV(e)$, and it will be omitted when this is obvious by the context or not relevant.
For any $e\in\Imp$, we have then $\pfunction{e}{}{}\ud \lambda x.\,{\tt out}(FV(e),\grasse{e}(\{{\tt input}(FV(e),x)\}))$. 

\subsection{Abstract interpretation} \label{cc-sec}

\subsubsection{Abstract domains.}
In standard Galois connection based abstract interpretation~\cite{CC77,CC79}, 
abstract domains (also called abstractions) are specified by Galois connections/insertions
(GCs/GIs for short). 
Concrete and abstract domains are usually complete lattices, resp.\ $\tuple{C,\leq_C}$ and
$\tuple{A,\leq_A}$,
which are related by abstraction and concretization maps
$\alpha:C\ra A$ and $\gamma:A \ra C$ that give
rise to a GC $(\alpha,C,A,\gamma)$ ($\tuple{\alpha,\gamma}$ for short), that is,
for all $a\in A$ and $c\in C$,
$\alpha(c) \leq_A a \Lra c \leq_C \gamma(a)$. 
A GC is a GI when $\alpha\circ\gamma=\lambda x.x$. In the following we assume GIs only.
Let us recall some basic properties of a GI
$(\alpha,C,A,\gamma)$: (1)~$\alpha$ is additive and  $\gamma$ is co-additive;  
(2)~$\gamma\circ \alpha: C\ra C$ is an upper closure operator, namely, it is 
a monotone, idempotent and increasing (i.e., $x \leq \gamma(\alpha(x))$) function; (3) if $\mu:C\ra C$ is an upper closure operator
then $(\mu,C,\mu(C),\lambda x.x)$ is a GI. 
If $\alpha:C\ra A$ is an additive function then it induces a GC $(\alpha,C,A,\alpha^+)$ 
where the concretization $\alpha^+:A \ra C$ is defined as right-adjoint of $\alpha$, i.e., 
$\alpha^+(a) \defi \vee_C \{c\in C  ~|~ \alpha(c) \leq_A a\}$. Dually, 
if $\gamma:C\ra A$ is a co-additive function then $(\gamma^-,C,A,\gamma)$ is a GC where
$\gamma^- \defi \lambda c. \wedge_A \{a\in A  ~|~ c \leq_C \gamma(a)\}$ is the left-adjoint of $\gamma$. 

\smallskip

We use $\Abs(C)$ to denote all the possible abstractions of a concrete domain $C$, where
$A\in \Abs(C)$ means that $A$ is an abstract domain of $C$ defined by some GI which is left unspecified, 
and $\auco{}{\cAA}$ is its corresponding upper closure operator.
If $A_1,A_2\in \Abs(C)$ then $A_1$ is equivalent to $A_2$, denoted by $A_1\sim A_2$, when 
$\gamma_{A_1}(A_1) = \gamma_{A_2}(A_2)$. The quotient  $\Abs(C)_{/\sim}$ is
called the lattice of abstractions because it turns out to be a complete lattice when $C$ is a cpo \cite{ran99}
w.r.t.\ the relative precision ordering: 
$A_1 \sqsubseteq A_2$ iff for any $c\in C$, $\gamma_{A_1}(\alpha_{A_1}(c)) \leq_C  \gamma_{A_2}(\alpha_{A_2}(c))$. 
$\Abs(C)_{/\sim}$ is indeed isomorphic to the lattice of upper closure operators $\uco(C)$ on the concrete lattice $\tuple{C,\leq_C}$.
Thus, $A_1 \sqsubseteq A_2$ means that $A_1$ is a more precise abstraction than $A_2$, or, equivalently, that $A_2$ abstracts $A_1$.  
An abstract domain 
$A\in \Abs(C)$ is called {\em trivial\/} when 
$A$ is either the least or the greatest abstract domain in $\tuple{\Abs(C)_{/\sim},\sqsubseteq}$, respectively $\auco{}{\cAA} = \lambda x.x$ (called the identical abstraction) or $\auco{}{\cAA} =
\lambda x.\top_C$, called the top abstraction and denoted $\cAA=\top$. We say that $\cAA$ is continuous when $\auco{}{\cAA}$ is continuous.


\subsubsection{Soundness and Completeness.}
Let $f:C\ra C$ be a
concrete monotone function---for simplicity 
we consider unary functions---and let
 $\ok{f^\sharp:A \ra A}$ be a corresponding monotone abstract function
 defined on  $\cAA\in \Abs(C)$ with GI $\tuple{\alpha,\gamma}$.
$\ok{f^\sharp}$ is a correct (or sound) approximation of $f$ on $A$
when $\ok{\alpha \circ f \sqsubseteq f^\sharp\circ \alpha}$ holds. 
If $\ok{f^\sharp}$ is correct for $f$ then it is fix-point correct, that is, 
$\alpha(\lfp(f)) \leq_A  \lfp(\ok{f^\sharp})$ holds. 
The abstract function
$\ok{f^A \ud \alpha \circ f \circ \gamma: A\rightarrow A}$ 
is called the {\em best
correct approximation\/} (bca for short) of $f$ on $A$, because it turns out that any abstract
function
$\ok{f^\sharp}$ is a correct approximation of $f$ iff $\ok{f^A \sqsubseteq f^\sharp}$.
Hence,
$\ok{f^A}$ plays the role of the 
best possible correct approximation of $f$ on the abstract domain $A$ \cite{CC79}. The tuple $\tuple{f,\cAA,\ok{f^\sharp}}$ fully defines an abstract interpretation. It is 
(globally) complete when  
$\ok{\alpha \circ f = f^\sharp \circ \alpha}$
\cite{CC79,GRS00}. When $\ok{f^\sharp}$ is an abstract transfer function on $A$ used by some static program analysis, completeness intuitively encodes an optimal precision for $\ok{f^\sharp}$, meaning 
that the abstract behavior of $\ok{f^\sharp}$ on $A$ 
exactly matches the abstraction in $A$ of the concrete behavior of $f$.  
If $\ok{f^\sharp}$ is complete for $f$ then
least fix-point completeness holds (also called fix-point transfer), i.e.,
$\alpha(\lfp(f)) =  \lfp(\ok{f^\sharp})$ holds. 
It turns out  that completeness $\ok{\alpha \circ f = f^\sharp \circ \alpha}$  holds  
iff $\ok{\alpha\circ f = (\alpha\circ f\circ\gamma)\circ \alpha = f^A \circ \alpha}$ holds.  
This corresponds precisely to ask that the following equation holds for all $x\in C$:
\begin{equation}\label{basic-completeness}
\auco{}{\cAA}(f(\auco{}{\cAA}(x)))=\auco{}{\cAA}(f(x)).
\end{equation}
When (\ref{basic-completeness}) holds only for a subset $S$ of the input domain, i.e., for all $x\in S\subseteq C$, then $\tuple{f,\cAA,\ok{f^\sharp}}$ is {\em locally complete} w.r.t.\ $S$ \cite{BruniGGR21}. Of course local completeness is in general a weaker property than completeness. 
Let us recall that function composition preserves completeness, that is, if $f$ and $g$ are complete on $A$ then $f\circ g$ is complete on $A$.
It has been recently proved that local completeness plays a key role in guaranteeing that the composition of bca's is still a bca \cite{GiacobazziR25}, i.e., 
$f_1^{\cAA}\comp f_2^{\cAA} = (f_1\comp f_2)^{\cAA}$.
In both cases, the possibility of
defining a complete approximation $\ok{f^\sharp}$ of $f$ on some $A\in \Abs(C)$ 
only depends upon the concrete function $f$ and on the abstraction $A$, that is, 
$\ok{f^A}$ is the only possible option as complete approximation of $f$. Therefore 
if $\ok{f^\sharp}$ is complete this implies that $\ok{f^\sharp=f^\cAA}$ while 
if $\ok{f^\cAA}$ is not complete then no abstract interpretation $\ok{f^\sharp}$ exists which can be complete.
%
%

\section{Abstract interpretation in standard recursion theory}\label{AbsIntRec}

Standard recursion theory is a theory of sets (properties) and functions indexed on elements of a programming system, i.e., programs. In the following we consider abstract interpretation in the context of partial computable functions. In this case $\cAA\in\Abs(\wpre(\Domain))$, meaning that we are only interested in abstractions of properties of programs, i.e., of properties $W_e\in \wpre(\Domain)$ for some $e\in\ccL$. 
When abstract interpretation is applied to program analysis or automated program verification, 
the abstraction is intended to associate with each program under inspection a decidable approximation of its often undecidable semantics. 
This is because the analysis of programs requires decidable answers to undecidable questions such as
those expressed as extensional properties of programs and concerning their dynamic behavior. 

The notion of decidable abstract interpretation has been studied in \cite{CGR18} for the comparison of the difficulty in analyzing and verifying programs. The notion of decidable abstraction in \cite{CGR18} specifies an abstract domain as a collection of recursive sets (properties) such that the relation of sub-set inclusion is decidable. We follow \cite{CGR18} and define $\cAA$ {\em decidable\/} when for all $e\in\ccL$, $\auco{}{\cAA}(W_e)$ is decidable. Note that among the trivial abstractions only $\cAA=\{\Domain\}$, i.e., $\auco{}{\cAA}=\lambda x.\Domain$, is decidable. $\auco{}{\cAA} = \lambda x.x$ is indeed clearly not decidable.
It is also worth noting that decidability does not imply that $\{\auco{}{\cAA}(W_e)\}_{e\in\ccL}$ is uniformly decidable. This would require that we can decide whether $x\in\auco{}{\cAA}(W_e)$ uniformly on both $x\in\Domain$ and $e\in\ccL$.
Next result proves that the set $\sset{(x,i)}{x\in\auco{}{\cAA}(W_i)}$ is decidable if and only if $\auco{}{\cAA}=\lambda x.\Domain$, i.e., $\cAA$ is the only trivial decidable abstraction.

\begin{theorem}\label{non-uniformrec}
    $\{\auco{}{\cAA}(W_e)\}_{e\in\ccL}$ is uniformly decidable if and only if $\auco{}{\cAA}=\lambda x.\Domain$.
\end{theorem}

\begin{proof}
    Assume $\{\auco{}{\cAA}(W_e)\}_{e\in\ccL}$ is uniformly decidable and $\auco{}{\cAA}\neq\lambda x.\Domain$. This means that $\sset{(x,i)}{x\in\auco{}{\cAA}(W_i)}$ is decidable and there exists $i\in\ccL$ and $\dot{x}\in\Domain$ such that $\dot{x}\not\in \auco{}{\cAA}(W_i)$. Consider the set $S_{\dot{x}}\ud\sset{i}{\dot{x}\not\in \auco{}{\cAA}(W_i)}\subseteq\ccL$. It is clear that $S_{\dot{x}}$ is an extensional property because if $e\in S_{\dot{x}}$ then $\dot{x}\not\in \auco{}{\cAA}(W_e)$ and if $\pfunction{e}{}{}\cong\pfunction{j}{}{}$ then $W_e=W_j$, therefore $\dot{x}\not\in \auco{}{\cAA}(W_j)$, i.e., $j\in S_{\dot{x}}$. By Rice's theorem $S_{\dot{x}}$ is decidable if and only if $S_{\dot{x}}=\varnothing$ or $S_{\dot{x}}=\ccL$, but none of these hold because: (1) 
    the program $*$ in Fig.~\ref{simple-code} is such that $W_*=\varnothing$ therefore because $\varnothing\subseteq W_i$ then by monotonicity $\auco{}{\cAA}(W_*)\subseteq \auco{}{\cAA}(W_i)$ which implies that
    $\dot{x}\not\in \auco{}{\cAA}(W_*)$ and therefore $*\in S_{\dot{x}}$, and (2) the program $p$ in Fig.~\ref{simple-code} is such that $W_p=\{\dot{x}\}$ which implies that $\dot{x}\in \auco{}{\cAA}(W_p)$, hence $p\not\in S_{\dot{x}}$. \myqed
    \begin{figure}[t]
    \begin{minipage}{6cm}
    \begin{center}
    \begin{align*}
    &1:~{\tt input}(x);\\
    &2:~{\tt while}\ {\tt true}\ \{ {\tt skip}\}.\\
    \end{align*}
    \end{center}
    \end{minipage}
    \begin{minipage}{6cm}
    \begin{center}
    \begin{align*}
    &1:~{\tt input}(x);\\
    &2:~{\tt if}\ x =\dot{x}\ {\tt then}\ {\tt skip}\\
    &3:~\qquad{\tt else}\ {\tt while}\ {\tt true}\ \{ {\tt skip}\}.
    \end{align*}
    \end{center}
    \end{minipage}
        \caption{The computable functions $\pfunction{*}{}{}$ and \pfunction{p}{}{}
        such that $W_*=\varnothing$ and $W_p=\{\dot{x}\}.$}\label{simple-code}
    \end{figure}
\end{proof}

We introduce the notion of a uniform closure operator as the effective counterpart of an upper closure operator: a transformation on program indices that realizes, at the level of programs, the closure induced by an abstract domain on the underlying semantic objects.

\begin{definition}
Let $r$ be a total computable function. $r$ is a uniform closure operator if for every index~$i\in\ccL$:
\begin{enumerate}
  \item $W_i \subseteq W_j ~\Rightarrow~W_{r(i)}\subseteq W_{r(j)}$,
  \item $W_i \subseteq W_{r(i)}$,
  \item $W_{r(i)}= W_{r(r(i))}$
\end{enumerate}
$r$ is decidable if for every index~$i\in\ccL$ we have that $W_{r(i)}$ is decidable.
\end{definition}

The following proposition proves that continuous and decidable abstract domains induce decidable uniform closure operators on $\ccL$.

\begin{proposition}
    Let $\cAA$ be decidable and continuous. Then $\{\auco{}{\cAA}(W_i)\}_{i\in\ccL}$ is a uniformly {\re} sequence and there exists a decidable uniform closure operator
    $r$ such that $\auco{}{\cAA}(W_i)=W_{r(i)}$ for all $i\in\ccL$.
\end{proposition}

\begin{proof}
    By continuity of $\auco{}{\cAA}$ we have that 
    \[
    \begin{array}{lll}
    \exists t.\; x\in\auco{}{\cAA}(\dovetl(i,t)) &\Leftrightarrow& x\in\cup_{t\in\omega}\auco{}{\cAA}(\dovetl(i,t))\\[1ex]
    &\Leftrightarrow& x\in\auco{}{\cAA}(\cup_{t\in\omega}\dovetl(i,t))\\[1ex]
    &\Leftrightarrow& x\in\auco{}{\cAA}(W_i).
    \end{array}
    \]
    Therefore we can define the following partial computable function:
    \[
    \pfunction{p}{}{}(i,x)=\left\{
    \begin{array}{cl}
       x &\mbox{if $\exists t.\, x\in\auco{}{\cAA}(\dovetl(i,t))$}\\[1ex]
       \uparrow & \mbox{otherwise}
    \end{array}\right .
    \]
    By the s-m-n theorem there exists a total computable function $g$ such that $\ok{\pfunction{g(p,i)}{}{}(x)=\pfunction{p}{}{}(i,x)}$. We set $\ok{r\ud\lambda i.\; g(p,i)}$. It is clear by continuity of $\auco{}{\cAA}$ that for all $i\in\ccL$: $W_{r(i)}=\sset{x}{\pfunction{p}{}{}(i,x){\downarrow}}=\sset{x}{\exists t.\; x\in\auco{}{\cAA}(\dovetl(i,t))}=\auco{}{\cAA}(W_i)$. Therefore by $\auco{}{\cAA}$ monotonicity, idempotency and extensivity we have $W_i\subseteq W_j~\Ra~\auco{}{\cAA}(W_i)\subseteq \auco{}{\cAA}(W_j)~\Leftrightarrow~W_{r(i)}\subseteq W_{r(j)}$, $W_i\subseteq \auco{}{\cAA}(W_i)=W_{r(i)}$ and 
    $W_{r(r(i))}=\auco{}{\cAA}(W_{r(i)})=\auco{}{\cAA}(\auco{}{\cAA}(W_{i}))=\auco{}{\cAA}(W_{i})=W_{r(i)}$.
    Decidability follows because all $W_{r(i)}$ are decidable. \myqed
\end{proof}

Uniformity plays a key role in viewing abstract interpretation within computability theory  because it expresses the fact that the abstraction function, in our case expressed by the closure operator $\auco{}{A}$, is itself defined by an algorithm (or a finite set of algorithms) that act uniformly and continuously over any input {\re} property. This is the case when $\auco{}{A}$ does not only encode the closure operation induced by a Galois insertion $\tuple{\alpha,\gamma}$, i.e., when $\auco{}{A}=\gamma\alpha$, but also the algorithm that computes the abstraction, i.e., that transforms a {\re} property $W_i$ into its abstraction $\auco{}{A}(W_i)$. In this sense, the hypothesis of continuity precisely corresponds to the effective computability of the abstraction operation. As observed above in Theorem~\ref{non-uniformrec}, if $r_{\cAA}$ is a decidable uniform closure operator associated with a decidable abstract domain $\cAA$, the uniform {\re} sequence $\{W_{r_{\cAA}(i)}\}_{i\in\ccL}$ cannot be uniformly decidable unless $r_{\cAA}$ maps any program $i\in\ccL$ into an always terminating program, i.e., $W_{r_{\cAA}(i)}=\Domain$, e.g., $r_{\cAA}(i)={\tt skip}$.

\smallskip

\begin{example}
    Consider the non trivial interval abstract domain $\mathsf{Int}$ \cite{CC77}, whose elements represent any property $S\in\wpre(\integer)$ of the integer values that a program variable $x$ may assume during the computation by the least interval $\mathsf{Int}(S)=[a,b]$ such that $S\subseteq [a,b]$, where $a\leq b$, $a\in\integer\cup\{-\infty\}$ and $b\in\integer\cup\{+\infty\}$, meaning that $x\in[a,b]$. This forms a well known decidable abstract domain enjoying a GI. We can associate with the corresponding decidable closure operator $\auco{}{\mathsf{Int}}$ a uniform program transformation ${\tt int}:\ccL\ra\ccL$ as follows:
    \[
    \pfunction{p}{}{}(i,x)=\left\{
    \begin{array}{ll}
      x & \mbox{if $\exists z,y\in W_i.\, z\leq x\leq y$}\\
      \uparrow & \mbox{otherwise}
    \end{array}
    \right .
    \]
    where by s-m-n theorem $\pfunction{{\tt int}(p,i)}{}{}(x)=\pfunction{p}{}{}(i,x)$. It is immediate to prove that for all $W_i\in\wpre(\integer)$, $\auco{}{\mathsf{Int}}(W_i)=W_{{\tt int}(p,i)}$, hence ${\tt int}$ is a uniform closure operator.
\end{example}

It is important to observe that the construction of the set $W_{{\tt int}(p,i)}$ corresponds precisely to dynamically checking the computed bounds of the program $i$. This is indeed achieved by sand-boxing, by the program transformation ${\tt int}$, the code of $i$ within $p$, as common practice in dynamic program analysis. This is the case of Daikon, a well-known tool for dynamic detection of approximate program invariants. Daikon infers likely invariants from executions, and these invariants can then be used for documentation, testing, or static checking \cite{Ernst-Daikon}.

It is obvious by Rice's theorem that the above uniform closure ${\tt int}$ is decidable but not uniformly decidable, i.e., the set $\sset{(x,i)}{x\in W_{{\tt int}(p,i)}}$ is in general not decidable (because for $x\in\Domain$ the set $\sset{i}{x\in W_{{\tt int}(p,i)}}$ is extensional). It is worth noting that if we replace $\auco{}{\cAA}(W_i)$ with any upper approximations of $W_i$ as induced by a decidable abstract interpretation defined on $\cAA$ (e.g., by employing widening operations on intervals \cite{CC77}) we cannot have as result a uniform closure operator. Denote the decidable abstract semantics defined by a terminating abstract interpretation on an abstract domain $\cAA$ as $\ok{\grasse{\cdot}^\cAA:\ccL\times\cAA\ra\cAA}$. It is clear that the set $\ok{\sset{(x,i)}{x\in\gamma(\grasse{i}^\cAA(\top_A))}}$ is now decidable, but any total computable function $s$ obtainable by s-m-n theorem and mapping $i\in\ccL$ into the index $s(i)$ such that $\ok{W_{s(i)}=\gamma(\grasse{i}^\cAA(\top_A))}$ cannot be in general a uniform closure. This is because by monotonicity this would imply that the abstract semantics $\ok{\grasse{i}^\cAA}$ induces an extensional equivalence relation on $\ccL$ and we know that this holds for trivial abstractions only \cite{BruniGGGP20}.
Consequently, a uniform closure operator can only arise from an upper closure operator $\auco{}{\cAA}$, namely uniformity is an intrinsic property of the abstract domain, not of any decidable abstract semantics that can be constructed over it. 

\smallskip

Because of these observations, throughout the following discussion, the decidable set $\auco{}{\cAA}(W_i)$, that corresponds to an element of the abstract domain $\cAA$, is always considered relative to a fixed and known property $W_i$. $W_i$ can be a decidable specification for some pre/post condition property or an intermediate property of the computation. The index $i$ will therefore only denote any known program representing that property.

\section{Classes of completeness and incompleteness}

In this section, we characterize the computational properties of the sets of programs
for which, given fixed abstract domain $\cAA$, the corresponding bca is locally complete. This naturally recalls the stronger notion of (global) completeness class firstly introduced in \cite{GiacobazziLR15} and corresponding to the class of all programs for which a given abstract domain $\cAA$ is globally complete: 
\[
\bbbc(\cAA) \defi \sset{e\in\ccL}{\forall i\in\ccL.\, 
\auco{}{A}(\pfunction{e}{}{}(W_i)) = \auco{}{A}(\pfunction{e}{}{}(\auco{}{A}(W_i)))}.
\]
It is worth noting that $\bbbc(\cAA)$, and hence also $\compl{\bbbc(\cAA)}$, are extensional properties of programs in $\ccL$, namely:
$e\in \bbbc(\cAA)~\wedge~\pfunction{e}{}{}\cong\pfunction{i}{}{}~\Rightarrow~i\in\bbbc(\cAA)$.
In \cite{GiacobazziLR15} Giacobazzi et al., proved that for a programming system $\ccL$ defined in terms of a simple imperative programming language {\em \`a la} $\Imp$ ::
\begin{itemize}
    \item $|\bbbc(\cAA)| = \aleph_0$;
    \item $\bbbc(\cAA)$ is decidable if and only if $A$ is trivial\footnote{Clearly when $A$ is trivial $\bbbc(\cAA)=\ccL$ and $\compl{\bbbc(\cAA)}=\varnothing$.};
    \item $A$ not trivial $\Rightarrow$ $\bbbc(\cAA)$ and $\compl{\bbbc(\cAA)}$ are productive sets, hence non {\re}.
\end{itemize}
By the fundamental isomorphism theorem \cite{Rogers58}, between any two acceptable programming systems there is an effective, one-to-one, and onto translation, hence because $\Imp$ is an acceptable programming system, any acceptable programming system satisfies the above three properties. In particular if $*\in\ccL$ and ${\tt id}\in\ccL$ denote respectively a program such that for all $x\in\Domain:\,\pfunction{*}{}{}(x){\uparrow}$ and $\pfunction{*}{}{}(x)=x$ (e.g., see Fig.~\ref{simple-code}), then $*,{\tt id}\in \bbbc(\cAA)$. 

Things are more complicated when we move from global completeness to the weaker notion of local completeness. 
We introduce the notion of local completeness class as the set of all programs for which its bca is locally complete with respect a fixed input {\re} property $W_i$.

\begin{definition}
    Let $\cAA$ be an abstract domain and $i\in \ccL$. The local completeness class of $\cAA$ with respect to $i$ is
    \[
\bbbc(\cAA,i) \defi \sset{e\in\ccL}{
\auco{}{A}(\pfunction{e}{}{}(W_i)) = \auco{}{A}(\pfunction{e}{}{}(\auco{}{A}(W_i)))}.
\]
\end{definition}
The local completeness class above naturally extends to sets of inputs as follows: let $S\subseteq\ccL$ then the local completeness class of $\cAA$ with respect to $S$ is
\[
\bbbc(\cAA,S)=\bigcap_{i\in S}\bbbc(\cAA,i).
\]
It is immediate to observe that if $S\subseteq R~\Rightarrow~\bbbc(\cAA,R)\subseteq\bbbc(\cAA,S)$, and for any $S\subseteq\ccL$: $\bbbc(\cAA)\subseteq\bbbc(\cAA,S)$, and $\bbbc(\cAA)=\bbbc(\cAA,\ccL)$.

\begin{proposition}\label{basic-prop}
Let $\cAA\in\Abs(\wpre(\ccD))$ be an abstract domain and
$S\subseteq\ccL$.
\begin{enumerate}[(i)]
\item 
$\bbbc(\cAA,S)$ is an extensional property;
\item
for all $S\subseteq\ccL$: $|\bbbc(\cAA,S)|=\aleph_0$.
\end{enumerate}
\end{proposition}
\begin{proof}
Extensionality in {\em (i)} is obvious.
{\em (ii)} $|\bbbc(\cAA,i)|=\aleph_0$ follows by extensionality in {\em (i)} and 
a straightforward padding argument because $\bbbc(\cAA)\neq\varnothing$ and hence, for any $S\subseteq\ccL$ also $\bbbc(\cAA,S)\neq\varnothing$. Indeed for any  
$W_i\in\wpre(\Domain)$: ${\tt id}\in\bbbc(\cAA)$ and ${\tt id}\in\bbbc(\cAA,i)$, where $\pfunction{{\tt id}}{}{}=\lambda x.\; x$. \myqed
\end{proof}

The proof of the following theorem follows the structure of the proof of a similar result for the global completeness case in \cite{GiacobazziLR15} and immediately extends that result to the case of local completeness.

\begin{theorem}\label{Decidable-C}
    Let $\cAA$ be an abstract domain and $i\in\ccL$ such that $W_i\neq\auco{}{A}(W_i)$.
    Then $\bbbc(\cAA,i)$ is decidable if and only if $\cAA=\top$.
\end{theorem}

\begin{proof}
    Being $\bbbc(\cAA,i)$ extensional, by Rice's theorem $\bbbc(\cAA,i)$ is decidable if and only $\bbbc(\cAA,i)=\varnothing$ or $\bbbc(\cAA,i)=\ccL$. Of course $\bbbc(\cAA,i)\neq\varnothing$, therefore it can only be $\bbbc(\cAA,i)=\ccL$. Moreover $W_i\neq\auco{}{A}(W_i)$ implies that $\cAA$ is trivial if and only if $\cAA=\top$. It is obvious that when $\cAA=\top$ then $\bbbc(\cAA,i)=\ccL$. Assume $\cAA\neq\top$ and $\bbbc(\cAA,i)=\ccL$. Then there exists $W_o\in\wpre(\Domain)$ such that $\auco{}{A}(W_o)\subsetneq\Domain$. Let $b\in\Domain\setminus\auco{}{A}(W_o)$ and $a\in \auco{}{A}(W_i)\setminus W_i$. We can define a partial computable function $\pfunction{e}{}{}$ computable by a program $e\in\ccL$ as follows:
    $$\pfunction{e}{}{}(x)\ud \left\{
    \begin{array}{ll}
       b & \mbox{if $x=a$}\\
       \uparrow &\mbox{otherwise}
    \end{array}
    \right .
    $$
    In this case $\pfunction{e}{}{}(W_i)=\varnothing$ and 
    $\pfunction{e}{}{}(\auco{}{A}(W_i))=\{b\}$. Note that 
    $\auco{}{A}(\varnothing)\neq \auco{}{A}(\{b\})$. This because 
    otherwise by $\auco{}{A}$ monotonicity: $\{b\}\subseteq \auco{}{A}(\{b\})=\auco{}{A}(\varnothing)\subseteq \auco{}{A}(W_o)$ would imply that 
    $b\in \auco{}{A}(W_o)$, which is a contradiction. Therefore $e\not\in\ccL$. \myqed
\end{proof}

An immediate consequence of this result is that, similar to the case of global completeness, for any input $W_i$ that is not precisely represented in the abstract domain, i.e., such that $W_i\neq\auco{}{A}(W_i)$, $\bbbc(\cAA,i)$ and its complement $\compl{\bbbc(\cAA,i)}$ are decidable if and only if $\cAA=\top$. Moreover, because $\bbbc(\cAA,i)$ is extensional, 
there always exist infinitely many 
programs for which the given abstraction is locally incomplete.  
In view of this observation, because an abstract interpretation is (locally) complete if and only if the corresponding bca is complete (see \cite{GRS00}[Lemma 3.1]), any abstract interpretation built on an abstraction $\cAA$ is locally complete for all programs if and only if $\cAA=\top$.

\medskip 

Theorem~\ref{Decidable-C} implies that when $W_i\neq\auco{}{A}(W_i)$ and $\cAA\neq\top$ (which implies that $\cAA$ is not trivial) both $\bbbc(\cAA,i)$ and $\compl{\bbbc(\cAA,i)}$ cannot be acceptable programming systems. On one side its is obvious that $*,{\tt id}\in \bbbc(\cAA,i)$. On the other side Theorem~\ref{Decidable-C} proved that any program 
$e\in\ccL$ such that
\begin{equation}\label{incomplete-prog}
    \pfunction{e}{}{}(x)\ud \left\{
    \begin{array}{ll}
       b & \mbox{if $x=a$}\\
       \uparrow &\mbox{otherwise}
    \end{array}
    \right .
\end{equation}
with $a\in \auco{}{A}(W_i)\setminus W_i$ and $b\in\Domain\setminus\auco{}{A}(W_o)$, we have $e\in\compl{\bbbc(\cAA,i)}$. 
Consider the following total computable function:
\begin{equation}\label{selection}
\pfunction{s}{}{}(a,b,c,x)\ud \left\{
    \begin{array}{ll}
       b & \mbox{if $x=a$}\\
       c &\mbox{otherwise}
    \end{array}
    \right .
\end{equation}
When $a,b,c$ are fixed as above, by the s-m-n theorem we have that there exists a total computable function $g$ such that $\pfunction{g(s,a,b,c)}{}{}(x)\cong\pfunction{s}{}{}(a,b,c,x)$.
It is immediate to prove by the same argument above that 
if $c\in W_o$ (i.e., $W_o\neq\varnothing$) then 
$g(s,a,b,c)\in\compl{\bbbc(\cAA,i)}$. 
This implies that both $\bbbc(\cAA,i)$ and $\compl{\bbbc(\cAA,i)}$ do not satisfy the axioms of Wagner's Uniform Reflexive Structures (URS) \cite{Wagner69} that model the essential properties of acceptable g{\"o}delizations, hence of acceptable programming systems. In particular for $\compl{\bbbc(\cAA,i)}$ we have that $*\not\in \compl{\bbbc(\cAA,i)}$ hence $\compl{\bbbc(\cAA,i)}$ violates Axiom 1 of URS requiring that the always undefined function has to be represented in any URS. Moreover for $\bbbc(\cAA,i)$ we have that if $a\in \auco{}{A}(W_i)\setminus W_i$, $b\in\Domain\setminus\auco{}{A}(W_o)$ and $c\in W_o$ then $g(s,a,b,c)\not\in \bbbc(\cAA,i)$ hence $\bbbc(\cAA,i)$ violates Axiom 2 of URS requiring that the function $\lambda a,b,c.\,g(s,a,b,c)$ performing selection for all $a,b,c\neq *$ has to be represented in any URS. For exactly the same reasons also $\compl{\bbbc(\cAA)}$ and $\bbbc(\cAA)$ cannot be acceptable programming systems. This observation implies that no matter if global or local completeness is concerned, neither completeness nor incompleteness classes can be considered as acceptable programming systems (hence languages) for partial computable functions, this independently from the equational theory used to specify the function in \ref{selection}, unless for the cases of $\bbbc(\cAA,i)$ and $\bbbc(\cAA)$ only, $\cAA$ is  trivial, i.e., $\bbbc(\cAA,i)=\bbbc(\cAA)=\ccL$.

\subsection{The hardness of local completeness}
As in the case of global completeness, we can go beyond and prove that $\bbbc(\cAA,i)$ is indeed a bi-productive set under some hypothesis on $i\in\ccL$ and $\cAA$.

\begin{theorem}\label{productivity}
    Let $\cAA$ be an abstract domain and $i\in\ccL$ such that $\cAA\neq\top$, $W_i\neq\auco{}{A}(W_i)$ (i.e., $\cAA$ is not trivial), and $W_i\neq\varnothing$. Then $\bbbc(\cAA,i)$ is a bi-productive set.
\end{theorem}
\begin{proof}
    To prove that $\bbbc(\cAA,i)$ is a bi-productive set we prove that $\compl{K}\preceq \bbbc(\cAA,i)$ and $\compl{K}\preceq \compl{\bbbc(\cAA,i)}$. As above, because $W_i\neq\auco{}{A}(W_i)$ and $\cAA\neq\top$ there exist $a\in \auco{}{A}(W_i)\setminus W_i$ and $W_o\in\wpre(\Domain)$ such that $b\in\Domain\setminus\auco{}{A}(W_o)$.
    
    For the case $\compl{K}\preceq \bbbc(\cAA,i)$ we can define a partial computable function $\pfunction{e}{}{}$ computable by a simple dovetail argument by a program $e\in\ccL$ as follows:
    $$\pfunction{e}{}{}(x,y)\ud \left\{
    \begin{array}{ll}
       b & \mbox{if  $x\in K~\wedge~y=a$}\\
       \uparrow &\mbox{otherwise}
    \end{array}
    \right .
    $$
    By the b-m-n theorem, there exists a total computable function $s:\ccL\rightarrow\ccL$ such that for any $x\in\ccL$ and $y\in\Domain$, 
    $\pfunction{e}{}{}(x,y)=\pfunction{s(e,x)}{}{}(y)$. We prove that $K\preceq_{s}\compl{\bbbc(\cAA,i)}$.
    \begin{itemize}
        \item 
        If $x\in K$ then $\pfunction{s(e,x)}{}{}(y)=b\Leftrightarrow y=a$. Because $\pfunction{s(e,x)}{}{}(W_i)=\varnothing$ then by monotonicity $\auco{}{A}(\varnothing)\subseteq\auco{}{A}(\{b\})$. Assume $\auco{}{A}(\varnothing)=\auco{}{A}(\{b\})$, then
        \[
        \{b\}\subseteq\auco{}{A}(\{b\})=\auco{}{A}(\varnothing)
        \subseteq \auco{}{A}(W_o)
        \]
        which implies $b\in \auco{}{A}(W_o)$ which is absurd. Therefore $s(e,x)\in \compl{\bbbc(\cAA,i)}$.\\
        \item 
        If $x\not\in K$ then for all $y\in\Domain$:
        $\pfunction{s(e,x)}{}{}(y){\uparrow}$, hence $s(e,x)\in \bbbc(\cAA,i)$.
    \end{itemize}

    For the case $\compl{K}\preceq \compl{\bbbc(\cAA,i)}$ we define, by a similar dovetail argument, a partial computable function $\pfunction{e}{}{}$ as follows:
    $$\pfunction{e}{}{}(x,y)\ud \left\{
    \begin{array}{ll}
       b & \mbox{if  $x\in K~\vee~y=a$}\\
       \uparrow &\mbox{otherwise}
    \end{array}
    \right .
    $$
    Again by the s-m-n theorem, there exists $s:\ccL\rightarrow\ccL$ a total computable function such that for any $x\in\ccL$ and $y\in\Domain$, 
    $\pfunction{e}{}{}(x,y)=\pfunction{s(e,x)}{}{}(y)$. We prove that $K\preceq_{s}\bbbc(\cAA,i)$. 
    \begin{itemize}
        \item 
        If $x\in K$ then 
        for all $y\in\Domain$:
        $\pfunction{s(e,x)}{}{}(y)=b$. Therefore, because $W_i\neq\varnothing$: $\pfunction{s(e,x)}{}{}(W_i)=\{b\}=\pfunction{s(e,x)}{}{}(\auco{}{A}(W_i))$ again implies that $s(e,x)\in \bbbc(\cAA,i)$.\\
        \item 
        If $x\not\in K$ then $\pfunction{s(e,x)}{}{}(y){\downarrow}\Leftrightarrow y=a$. Because $a\not\in W_i$ then $\pfunction{s(e,x)}{}{}(W_i)=\varnothing$ and $\ok{\pfunction{s(e,x)}{}{}(\auco{}{A}(W_i))=\{b\}}$ which implies that $s(e,x)\in \compl{\bbbc(\cAA,i)}$ by the same argument above.
    \end{itemize} \myqed
\end{proof}
It is worth noting that the productivity of $\bbbc(\cAA,i)$ does not require $W_i\neq\varnothing$, which is instead required in the proof of case $\ok{\compl{\bbbc(\cAA,i)}}$. This is because if $\auco{}{A}(\varnothing)\in\wprec(\Domain)$ (e.g., when $\cAA$ is decidable) and $W_i=\varnothing$ the set $\ok{\compl{\bbbc(\cAA,i)}}$ is indeed {\re}. 

\begin{proposition}
    Let $W_i=\varnothing$ and $\auco{}{A}(\varnothing)\neq\varnothing$. If $\auco{}{A}(\varnothing)\in\wprec(\Domain)$ then $\compl{\bbbc(\cAA,i)}$ is a creative set.
\end{proposition}

\begin{proof}
    By Theorem~\ref{productivity} $K\preceq \compl{\bbbc(\cAA,i)}$, it is therefore sufficient to prove that $\compl{\bbbc(\cAA,i)}$ is {\re}. Given any $W_j\in\wpre(\Domain)$ we first prove that:
    \begin{equation}\label{usefuleq}
        \auco{}{A}(W_j)\neq\auco{}{A}(\varnothing)~\Leftrightarrow~W_j\not\subseteq\auco{}{A}(\varnothing)
    \end{equation}
    ($\Leftarrow$) Assume $W_j\not\subseteq \auco{}{A}(\varnothing)$ then let $c\in W_j\setminus \auco{}{A}(\varnothing)$. By $\auco{}{A}$ extensivity $c\in\auco{}{A}(W_j)$ and $c\not\in\auco{}{A}(\varnothing)$ hence $\auco{}{A}(\varnothing)\subset \auco{}{A}(W_j)$ which implies $\auco{}{A}(W_j)\neq\auco{}{A}(\varnothing)$. \\
    ($\Rightarrow$) Assume $W_j\subseteq \auco{}{A}(\varnothing)$ then by monotonicity and idempotency we have $\auco{}{A}(W_j)\subseteq\auco{}{A}(\varnothing)$. Because $\varnothing\subseteq W_j$ then by monotonicity $\auco{}{A}(\varnothing)\subseteq \auco{}{A}(W_j)$, which implies that $\auco{}{A}(\varnothing)=\auco{}{A}(W_j)$. 

    We consider the following partial computable function
    \[
    \pfunction{q}{}{}(a,b,y)\ud\left\{
    \begin{array}{ll}
        1 & \mbox{of $\exists z\in W_b.~\pfunction{a}{}{}(z){\downarrow}~\wedge~y=\pfunction{a}{}{}(z)$}\\
        \uparrow &\mbox{otherwise}
    \end{array} \right.
    \]
    If $e\in\ccL$ and $\auco{}{A}(\varnothing)=W_w\in\wprec(\Domain)$ for some index $w\in\ccL$, by the s-m-n theorem there exists a total computable function $s:\ccL\ra\ccL$ such that 
    $\pfunction{q}{}{}(e,w,y)=\pfunction{s(q,e,w)}{}{}(y)$ and by definition:
    \[
    \begin{array}{lll}
    W_{s(q,e,w)} &= & \sset{y}{\pfunction{s(q,e,w)}{}{}(y){\downarrow}}\\[1ex]
    &=& \sset{y}{\exists z\in W_w.\, \pfunction{e}{}{}(z){\downarrow}~\wedge~y=\pfunction{e}{}{}(z)}\\
    &=& \pfunction{e}{}{}(W_w).
    \end{array}
    \]
    Therefore whenever $W_i=\varnothing$ by (\ref{usefuleq}) we have $\compl{\bbbc(\cAA,i)}=\sset{e}{W_{s(q,e,w)}\not\subseteq W_w}$, which is a clearly {\re} set because $\auco{}{A}(\varnothing)=W_w\in\wprec(\Domain)$. \myqed
\end{proof}


\subsection{$\bbbc(\cAA,i)$ and $\compl{\bbbc(\cAA,i)}$ in the arithmetical hierarchy}

Although abstract interpretation theory (\cite{CC79}) is defined independently from the assumption of dealing with {\re} properties, this assumption is implicit in most of these works for the obvious reason that when applied to program analysis and verification, properties concern the semantics of programs, which are by definition {\re} sets. With this assumption we were able to prove that when the input property does not represent non-terminating computations, i.e., $W_i\neq\varnothing$ and $\cAA$ is not trivial then $\bbbc(\cAA,i)$ is bi-productive. This gives a lower bound in the arithmetic hierarchy for $\bbbc(\cAA,i)$ and its complement $\compl{\bbbc(\cAA,i)}$. In particular if we consider the predicate:
\[
\predC(e,i) \iff \auco{}{A}(\pfunction{e}{}{}(\auco{}{A}(W_i))) = \auco{}{A}(\pfunction{e}{}{}(W_i))
\]
then clearly $\predC(e,i)$ sits at level $\geq 2$ in the arithmetical hierarchy, i.e., $\predC(e,i)\in \Sigma_{n}^0\cup\Pi_{n}^0$ with $n\geq 2$. Instead, when $W_i=\varnothing$ we proved that $\compl{\bbbc(\cAA,i)}$ is a creative set, i.e., $\compl{\bbbc(\cAA,i)}$ is {\re} and $\bbbc(\cAA,i)$ is productive (co-{re}). This means that when $W_i=\varnothing$ then $\compl{\bbbc(\cAA,i)}\in\Sigma_1^0$ and $\bbbc(\cAA,i)\in\Pi_1^0$.

In order to prove more properties about local completeness classes, such as a tightest inclusion in the arithmetic hierarchy or the possibility of having a {\re} representation of the productive set $\bbbc(\cAA,i)$, 
we need more structure on the construction of the abstraction function, such that the fact that $\sset{\auco{}{\cAA}(W_i)}{i\in\ccL}$ is a uniformly {\re} sequence of recursive sets, as proved in Section \ref{AbsIntRec}.



Assume $\cAA$ be decidable and continuous with decidable corresponding uniform closure $r_{\cAA}$ and 
$h:\ccL\ra\ccL$ a total computable function such that 
    \[
     W_{h(e,i)} \;=\; \pfunction{e}{}{}(W_{i})
        \;=\; \sset{\pfunction{e}{}{}(x)}{x \in W_{i} \;\land\; \pfunction{e}{}{}(x){\downarrow}}.
    \]
In this case:
\[
  \predC_\cAA(e,i) \;\iff\; W_{r_{\cAA}(h(e,\,r_{\cAA}(i)))} \subseteq W_{r_{\cAA}(h(e,\,i)))}.
\]
It is clear that given a fixed index $i\in\ccL$:
\begin{align*}
  \bbbc(\cAA,i) \;&=\; \sset{e}{\predC(e,i)},\\[1ex]
\compl{\bbbc(\cAA,i)} \;&=\; \sset{e}{\neg\predC(e,i)},\\[1ex]
  \bbbc(\cAA) \;&=\; \sset{e}{\forall i.\; \predC(e,i)}, \\[1ex]
  \compl{\bbbc(\cAA)} \;&=\; \sset{e}{\exists i.\; \neg\predC(e,i)}.
\end{align*}

\begin{lemma}\label{hierarchy}
    If $\cAA$ is decidable and continuous then $\predC_\cAA(e,i) \in \Pi^0_2$.
\end{lemma}

\begin{proof}
The inclusion $W_{r_{\cAA}(h(e,\,r_{\cAA}(i)))} \subseteq W_{r_{\cAA}(h(e,\,i)))}$ corresponds to:
\[
  \predC_\cAA(e,i) \;\iff\;
  \forall x.\,\bigl(\pfunction{r_{\cAA}(h(e,r_{\cAA}(i)))}{}{}(x){\downarrow}
  \;\Ra\;
  \pfunction{r_{\cAA}(h(e,i))}{}{}(x){\downarrow}\bigr).
\]
By replacing convergence with Kleene's $\NF$-predicate we have:
\[
  \predC(e,i) \;\iff\;
  \forall x.\,\Bigl(
    \bigl(\forall s.\;\neg \NF\bigl(r_{\cAA}(h(e,r_{\cAA}(i))),\, x,\, s\bigr)\bigr)
    \;\lor\;
    \bigl(\exists t.\; \NF\bigl(r_{\cAA}(h(e,i)),\, x,\, t\bigr)\bigr)
  \Bigr).
\]
Then we obtain a prenex form as follows:
\[
  \predC(e,i) \;\iff\;
  \forall x.\;\forall s.\;\exists t.\;
  \Bigl(\neg \NF\bigl(r_{\cAA}(h(e,r_{\cAA}(i))),\, x,\, s\bigr)
  \;\lor\;
  \NF\bigl(r_{\cAA}(h(e,i)),\, x,\, t\bigr)\Bigr).
\]
By collapsing the two adjacent universal quantifiers via the pairing function we have:
\[
  \predC(e,i) \;\iff\;
  \forall \langle x, s \rangle.\;\exists t.\;
  R(e,\, i,\, x,\, s,\, t)
\]
where $R$ is a decidable predicate.  This proves that $\predC(e,i) \in \Pi^0_2$. \myqed
\end{proof}

It is worth noting that the properties of decidability, continuity and non triviality are shared by the abstract domains of most of the Galois-connection based abstract interpretations \cite{CGR18}.
The proof of the following theorem is therefore immediate by Lemma \ref{hierarchy} and it places all the non-trivial completeness and incompleteness classes into the arithmetic hierarchy.

\begin{theorem}
    If $\cAA$ is decidable, continuous and not trivial such that $W_i\neq\varnothing$ then $\bbbc(\cAA,i)\in\Pi^0_2$, $\bbbc(\cAA)\in\Pi^0_2$, $\compl{\bbbc(\cAA,i)}\in\Sigma^0_2$, and $\compl{\bbbc(\cAA)}\in\Sigma^0_2$.
\end{theorem}

\section{Computable coverings and program transformations}

The interesting consequence of decidable abstract domains---as those employed in static program analysis, is that it is possible to algorithmically extract a {\re} (and therefore also a decidable) covering of the productive set $\bbbc(\cAA,i)$. In general, this cannot be achieved for \(\compl{\bbbc(\cAA,i)}\). The reason is that, for every program \(e \in \ccL\), there may be a potentially infinite family of effective transformations mapping $e$ to programs $f(e)$ such that \(f(e) \in \compl{\bbbc(\cAA,i)}\). Moreover, different choices of $f(e)$ may compute different functions $\pfunction{f(e)}{}{}$, so there is no canonical effective transformation that uniformly represents all such cases. We show that there exists a transformation that employs a constant growth in program size and minimal computational complexity.

\subsection{Computable coverings of $\bbbc(\cAA,i)$}
For $\bbbc(\cAA,i)$ it is possible to algorithmically ``cancel'' the outputs causing incompleteness from any partial computable function, yet obtaining an effective enumeration of a set of programs that represents all partial computable functions computable by programs in $\bbbc(\cAA,i)$.  
This is obtained by removing from the range of $\pfunction{e}{}{}$ all those outputs of $e$ that, coming from the approximation of $W_i$ according to $\auco{}{\cAA}$, namely from $\auco{}{\cAA}(W_i)$, end up outside the approximation of the range of $e$ on input $W_i$, i.e., 
outside $\auco{}{\cAA}(\pfunction{e}{}{}(W_i))$.
We prove that this program transformation is effective, therefore producing a {\re} cover for $\bbbc(\cAA,i)$.

\begin{theorem}\label{covering}
    If $\cAA$ is decidable and continuous then $\bbbc(\cAA,i)$ admits a {\re} covering, i.e., there exists $\cC_i\subseteq \bbbc(\cAA,i)$ such that $\cC_i$ is {\re} and $\cC_i$ and $\bbbc(\cAA,i)$ are extensionally equivalent.
\end{theorem}

\begin{proof}
    For any $i,e\in\ccL$ by the s-m-n theorem we have that there exists a total computable function $h:\ccL\ra\ccL$ such that
    \[
     W_{h(e,i)} \;=\; \pfunction{e}{}{}(W_{i})
        \;=\; \sset{\pfunction{e}{}{}(x)}{x \in W_{i} \;\land\; \pfunction{e}{}{}(x){\downarrow}}.
    \] 
    Fix $i\in\ccL$ and assume that $\auco{}{A}$ is a decidable closure operator, i.e., 
    $\auco{}{\cAA}(W_i)$ is a fixed decidable set. Define the partial computable function
    \begin{equation}\label{covering-function}
    \pfunction{q}{}{}(e,x)=\left\{
    \begin{array}{cl}
       \pfunction{e}{}{}(x) &\mbox{if $x\not\in \auco{}{A}(W_{i})$}\\[1ex]
       \pfunction{e}{}{}(x) &\mbox{if $x\in \auco{}{A}(W_{i})~\wedge~\pfunction{e}{}{}(x){\downarrow}~\wedge~\pfunction{e}{}{}(x)\in \auco{}{A}(W_{h(e,i)})$}\\[1ex]
       \uparrow & \mbox{otherwise}
    \end{array}\right .
    \end{equation}
    Because when $i\in\ccL$ is fixed $\auco{}{A}(W_{i})$ is decidable and $\auco{}{A}$ is monotone and extensive, then the function $\pfunction{q}{}{}$ is computable by the program $q$ in Fig.~\ref{code1}. In Fig.~\ref{code1} we abuse notation and denote with $e(x)$ the code of $e\in\ccL$ specialized for the input $x$.
    \begin{figure}[t]
    \begin{align*}
    &1:~{\tt input}(e,x);\\
    &2:~{\tt if}\ x \notin \auco{}{A}(W_i)\ {\tt then}\ e(x)\\
    &3:~\qquad\! {\tt else}\ n := e(x);\\
    &4:~\qquad\qquad W := \varnothing;\ t := 0;\\
    &5:~\qquad\qquad {\tt while}\ n \notin W\ \{\\
    &6:~\qquad\qquad\qquad t := t + 1;\\
    &7:~\qquad\qquad\qquad W := \dovetl(h(e,i),t);\\
    &8:~\qquad\qquad\qquad W := \auco{}{\mathcal{A}}(W)\\
    &9:~\qquad\qquad \};\ e(x).
    \end{align*}
    \caption{An implementation of $\pfunction{q}{}{}$.}\label{code1}
    \end{figure}
    Because $\auco{}{A}$ is monotone and for all $t\in\omega$:
    $\dovetl(h(e,i),t)\subseteq W_{h(e,i)}$ and $\dovetl(h(e,i),t)\subseteq \dovetl(h(e,i),t+1)$
    then for all $t\in\omega$
    $$\auco{}{\cAA}(\dovetl(h(e,i),t))\subseteq \auco{}{\cAA}(\dovetl(h(e,i),t+1))\subseteq \auco{}{A}(W_{h(e,i)}).$$ 
    When $\pfunction{e}{}{}(x){\downarrow}$ and $n=\pfunction{e}{}{}(x)$ (i.e., we can proceed with program line 4 in Fig.~\ref{code1}), by $\auco{}{\cAA}$ continuity and $\dovetl$ definition, we have that the {\tt while}-loop at line 5 in Fig.~\ref{code1} terminates if and only if: 
    \[
    \begin{array}{lll}
    \exists t.\; n\in\auco{}{\cAA}(\dovetl(h(e,i),t)) &\Leftrightarrow&\pfunction{e}{}{}(x)\in\cup_{t\in\omega}\auco{}{\cAA}(\dovetl(h(e,i),t))\\[1ex]
    &\Leftrightarrow&\pfunction{e}{}{}(x)\in\auco{}{\cAA}(\cup_{t\in\omega}\dovetl(h(e,i),t))\\[1ex]
    &\Leftrightarrow&\pfunction{e}{}{}(x)\in\auco{}{\cAA}(\pfunction{e}{}{}(W_i)).
    \end{array}
    \]
    Therefore if $\pfunction{e}{}{}(x)\in\auco{}{\cAA}(\pfunction{e}{}{}(W_i))$ then we have
    $\pfunction{q}{}{}(e,x)=\pfunction{e}{}{}(x)$.\\
    If instead for all $t\in\omega: n\not\in\auco{}{\cAA}(\dovetl(h(e,i),t))$ then $\pfunction{q}{}{}(e,x){\uparrow}$.
    Note that when $\pfunction{q}{}{}(e,x){\downarrow}~\Ra\pfunction{e}{}{}(x){\downarrow}$. The reversal does not hold because $\pfunction{q}{}{}(e,x)$ cancels (being undefined) all the outputs that may generate incompleteness. 
    
    By the s-m-n theorem there exists a total computable function $g:\ccL\ra\ccL$ such that for any $e\in\ccL$: $\pfunction{q}{}{}(e,x)=\pfunction{g(q,e)}{}{}(x)$. We set $\cC_i\ud \sset{g(q,e)}{e\in\ccL}$. $\cC_i$ is clearly {\re} being the range of a total computable function. We prove the following two facts: (1) $\cC_i\subseteq \bbbc(\cAA,i)$ and (2) $\sset{\pfunction{e}{}{}}{e\in\cC_i}=\sset{\pfunction{e}{}{}}{e\in\bbbc(\cAA,i)}$.
    \begin{enumerate}
        
    
        \item We prove that for all $e\in\ccL$: $g(q,e)\in \bbbc(\cAA,i)$, namely:
        \[
            \auco{}{A}(\pfunction{g(q,e)}{}{}(W_i)) 
                = \auco{}{A}(\pfunction{g(q,e)}{}{}(\auco{}{A}(W_i))).
        \]
        If \(x \in \auco{}{A}(W_i)\), then either \(\varphi_{g(q,e)}(x) = \varphi_e(x)\) and \(\varphi_e(x) \in 
        \auco{}{A}(W_{h(e,i)})\) or \(\pfunction{g(q,e)}{}{}(x)\!\uparrow\). Consequently,
        \begin{align*}
            \pfunction{g(q,e)}{}{}(\auco{}{A}(W_i))
                &= \{\varphi_e(x) \;\vert\; x \in \auco{}{A}(W_i) 
                    \land \varphi_e(x) \in \auco{}{A}(W_{h(e,i)})\} \\
            &= \varphi_e(\auco{}{A}(W_i)) \cap \auco{}{A}(W_{h(e,i)}) \\
            &= \varphi_e(\auco{}{A}(W_i)) \cap \auco{}{A}(\varphi_e(W_i)).
        \end{align*}
        By monotonicity of \(\auco{}{A}\): \(\varphi_e(W_i) \subseteq
        \auco{}{A}(\varphi_e(W_i))\). Additionally, \(W_i \subseteq \auco{}{A}(W_i)\)
        implies \(\varphi_e(W_i) \subseteq \varphi_e(\auco{}{A}(W_i))\). It then holds
        that
        \[
            \varphi_e(W_i) 
                \subseteq \varphi_e(\auco{}{A}(W_i)) \cap \auco{}{A}(\varphi_e(W_i))
                \subseteq \auco{}{A}(\varphi_e(W_i)).
        \]
        Again, since \(\auco{}{A}\) is monotonic, order is preserved if \(\auco{}{A}\)
        is applied term-wise:
        \[
            \auco{}{A}(\varphi_e(W_i)) 
                \subseteq \auco{}{A}\big(\varphi_e(\auco{}{A}(W_i))
                    \cap \auco{}{A}(\varphi_e(W_i))\big)
                \subseteq \auco{}{A}(\auco{}{A}(\varphi_e(W_i))).
        \]
        Substituting equivalent terms and eliminating, by idempotence of \(\auco{}{A}\),
        the extra \(\auco{}{A}\) on the right, we obtain
        \[
            \auco{}{A}(\varphi_{g(q,e)}(W_i)) 
                \subseteq \auco{}{A}(\varphi_{g(q,e)}(\auco{}{A}(W_i)))
                \subseteq \auco{}{A}(\varphi_{g(q,e)}(W_i))
        \]
        which implies
        \[
            \auco{}{A}(\pfunction{g(q,e)}{}{}(W_i)) 
                = \auco{}{A}(\pfunction{g(q,e)}{}{}(\auco{}{A}(W_i))).
        \]
        \color{black}
        \item By definition $\cC_i$ covers $\bbbc(\cAA,i)$ if $e\in\bbbc(\cAA,i)~\Ra~\pfunction{g(q,e)}{}{}\cong\pfunction{e}{}{}$. If $x\not\in \auco{}{\cAA}(W_{i})$ then by definition
        $\pfunction{g(q,e)}{}{}\cong\pfunction{e}{}{}$ (line 2 in Fig.~\ref{code1}). If instead $x\in \auco{}{\cAA}(W_{i})$ then if $\pfunction{e}{}{}(x){\downarrow}$ then because $e\in\bbbc(\cAA,i)$ by $\auco{}{\cAA}$ extensivity we have:
        \[
        \pfunction{e}{}{}(x)\in\pfunction{e}{}{}(\auco{}{\cAA}(W_{i}))
        \subseteq \auco{}{\cAA}(\pfunction{e}{}{}(\auco{}{\cAA}(W_{i})))=
        \auco{}{\cAA}(\pfunction{e}{}{}(W_{i}))
        \]
        and therefore by definition of $\pfunction{g(q,e)}{}{}$ (the loop at line 5--8 in Fig.~\ref{code1} terminates) we have $\pfunction{g(q,e)}{}{}(x)=\pfunction{e}{}{}(x)$ (line 9 in Fig.~\ref{code1}). If instead $\pfunction{e}{}{}(x){\uparrow}$ then clearly 
        $\pfunction{g(q,e)}{}{}(x){\uparrow}$ and again $\pfunction{g(q,e)}{}{}\cong\pfunction{e}{}{}$.
    \end{enumerate}
    \myqed
\end{proof}

The theorem above says that the local completeness class $\bbbc(\cAA,i)$, when it is not straightforward, admits a {\re} cover when $\cAA$ is a decidable and continuous abstract domain. $\cC_i$ is a family of programs that fully represent the functions computed by the programs in $\bbbc(\cAA,i)$ and these programs can be effectively build by the total computable function $\lambda e.\,g(q,e):\ccL\ra\ccL$ that transforms any, possibly incomplete, program $e$ into a complete one $g(q,e)$ by canceling those outputs that produce incompleteness. Here decidability plays a key role because it allows us to decide whether any concrete state $x$ belongs to the abstract input property $\auco{}{\cAA}(W_i)$. Without this assumption the function $\pfunction{g(q,e)}{}{}$ is not computable and hence the map $e\mapsto g(q,e)$ does not exist.

\begin{example}
    It is easy to see how the program transformation $g$ in Theorem~\ref{covering} removes the sources of incompleteness. For example for the interval abstract domain with $W_i=\{0,10\}$ and $W_o=\{100\}$, for any program $e$ implementing the archetypal incomplete function (\ref{incomplete-prog}) used in Theorem~\ref{Decidable-C}, e.g.,
    the program 
     \begin{align*}
    &1:~{\tt input}(x);\\
    &2:~{\tt if}\ x = 7\ {\tt then}\ 100\\
    &3:~\qquad\! {\tt else}\ {\tt while}\ {\tt true}\ \{{\tt skip}\}.
    \end{align*}
    we have that $\pfunction{g(q,e)}{}{}=\pfunction{*}{}{}$ and $*\in\bbbc(\cAA,i)$. In this case this is due to the non termination of the {\tt while}-loop between lines 5--8 in Fig.~\ref{code1}, because $W_{h(e,i)}=\pfunction{e}{}{}(\{0,10\})=\varnothing$ and $\auco{}{{\tt Int}}(\varnothing)=\varnothing$.
\end{example}

Interestingly, Theorem \ref{covering} gives also the very first fix-point characterization of local completeness: $\bbbc(\cAA,i)=\sset{e}{\pfunction{e}{}{}\cong\pfunction{g(q,e)}{}{}}$. This immediately implies that $\bbbc(\cAA,i)\neq\varnothing$, $|\bbbc(\cAA,i)|=\aleph_0$, and  $\bbbc(\cAA,i)\in\Pi^0_2$. 

\smallskip

Because it is known that every {\re} set of indexes of partial computable functions $S$ admits a decidable subset $B\subseteq S$ which covers $S$ \cite{odifreddi}, and this set $B$ can be effectively enumerated, then the following corollary follows immediately.

\begin{corollary}
    In the hypothesis of Theorem~\ref{covering}, $\bbbc(\cAA,i)$ has a recursive cover.
\end{corollary}


\subsection{Making programs incomplete by regular translations}

We cannot in general build a {\re} cover for $\compl{\bbbc(\cAA,i)}$. However
Theorem \ref{Decidable-C} shows that 
incompleteness can occur in even the simplest of programs. We show that
every non-trivial abstraction is incomplete for some program as simple as a decider of a regular 
language---the simplest in Chomsky's hierarchy. 

To ensure a program is computationally as simple as a regular language decider, it sufficient to
construct the program in those very terms. In particular, we will construct a program from the
standard model for deciding regular languages, the deterministic finite automata (DFA). 
%
We use a standard formulation of DFAs. A DFA is a 5-tuple \(\langle S,Q,F,q_0,\delta\rangle\)
where \( S\) is a finite alphabet of symbols, \(Q\) a finite set of states, \(F \subseteq Q\) a
subset of states designated as `accepting', \(q_0 \in Q\) the start state, and 
\(\delta: (Q \times  S) \to Q\) the function that defines the state transition rules
A DFA can be defined by its characteristic function \(M:  S^* \to \{0,1\}\) where \( S^*\) is
the set of finite-length strings comprised of symbols in \( S\) with typical elements $\sigma$. The characteristic function and
the 5-tuple can be connected operationally via the helper function \(\hat{\delta}: ( S^* \times Q) \to Q\) defined recursively as usual:
\[
    \begin{cases}
        \hat{\delta}(\varepsilon, q) = q,\\
        \hat{\delta}(x \cdot \sigma, q) = \hat{\delta}(\sigma, \delta(x, q))
    \end{cases}
\]
so that \(M\) can be defined
\[
    M(\sigma) \defeq \begin{cases}
        1 & \hat{\delta}(\sigma, q_0) \in F,\\
        0 & \hat{\delta}(\sigma, q_0) \notin F.
    \end{cases}
\]
It is known that the set of regular languages is exactly the set of all languages \(L\subseteq S^*\) for which there exists
a DFA \(M\) such that for any string \(\sigma\), we have \(\sigma \in L\) if and only if \(M(\sigma) = 1\).

Let \(S\) be an alphabet of symbols used for encoding and \(\tau: \Domain \to  S^*\) a linear time string
encoding for $\Domain$, then a function \(f_M: \Domain^3 \to \Domain\) is a {\em regular function\/} if, for some DFA \(M\), it can be defined 
\[
    f_{M}(a,e,x) = \begin{cases}
        a & M(\tau(x)) = 1,\\
        e & M(\tau(x)) = 0
    \end{cases}
\]
A regular program is any program that implements a regular function.

In our context and from the construction of Theorem \ref{Decidable-C} we can define the following partial computable function as implemented in Fig. \ref{code3}:
\[
\pfunction{p}{}{}(a,b,e,x)=\left\{
\begin{array}{cl}
b & \mbox{if $x=a$}\\
\pfunction{e}{}{}(x) &\mbox{otherwise}
\end{array}\right .
\]
By the s-m-n theorem there exists a total computable function $g$ such that for fixed $a,b\in\Domain$ and $e\in\ccL$, $\pfunction{p}{}{}(a,b,e,x) = \pfunction{g(p,a,b,e)}{}{}(x)$. By what proved in Theorem \ref{Decidable-C} it is immediate to see that when $\cAA$ is not straightforward and $W_o\in\wpre(\Domain)$ such that $b\in\Domain\setminus\auco{}{A}(W_o)$ and $a\in \auco{}{A}(W_i)\setminus W_i$, then for all $e\in\ccL$ we have that $g(p,a,b,e)\in \compl{\bbbc(\cAA,i)}$. 
\begin{figure}[t]
    \begin{align*}
    &1:~{\tt input}(a,b,e,x);\\
    &2:~{\tt if}\ x = a\ {\tt then}\ b\\
    &3:~\qquad\qquad {\tt else}\ e(x).
    \end{align*}
    \caption{An implementation of $\pfunction{p}{}{}$.}\label{code3}
    \end{figure}
It is immediate to note that there exists a DFA $M_{L}$ that recognizes the regular language  $L=\{\tau(a)\}$ or more in general $L=\tau(R)$ for any regular language $R\subseteq \auco{}{\cAA}(W_i)\setminus W_i$, such that 
$\pfunction{g(p,a,b,e)}{}{}(x)= f_{M_{L}}(b,\pfunction{e}{}{}(x),x)$. The function $g$ mapping $e$ into $g(p,a,b,e)$ can therefore be implemented as a regular program, returning a transformed program whose running time has a minimal linear overhead over the complexity of $e$.

\section{Conclusion}

We studied local completeness of abstract interpretations in the standard computability context. In particular we analyzed the complexity of the classes of programs having a locally complete and incomplete abstract interpretation. We first reconciled abstract interpretation with standard computability theory by defining the notion of uniform closure operators as computable program transformations acting as semantic approximations of {\re} sets, and proved that any Galois connection-based abstract domain naturally induces a uniform closure. We then observed, via a simple application of Rice's theorem, that no such closures can correspond to a uniformly decidable abstraction unless the abstraction is trivial. This provides a precise computability-theoretic account of the familiar distinction between the static and dynamic use of abstractions.

Upon these bases we proved that the classes of locally complete and incomplete abstract interpretations are in general recursively inseparable being both productive sets. Moreover, although a {\re} covering cannot in general be obtained for the incompleteness case, it can be effectively constructed for the completeness case. This result opens a promising direction for future work on the effective identification of representative code blocks for which a given abstraction is complete/precise. In particular, our covering function provides an elegant fix-point characterization of the classes of locally complete programs, paving the road to systematic code repair strategies to improve the precision of abstract interpretations by program transformation. In particular, because local completeness is essential to make abstract interpreters compositional, this result can be used to solve the problem of decomposing a program $p$ into suitable subprograms $p_1,\ldots,p_n$ such that the composition of their abstract interpretations minimizes incompleteness, as suggested in \cite{GiacobazziR25}. Such blocks could guide targeted refinements toward the portions of code responsible for imprecision, enabling more precise analyses without requiring a global refinement of the abstract domain. We believe that a similar construction can be used to obtain a covering for the class of partially complete abstract interpretations as introduced in \cite{CDG22}. In this case the program transformation can be used to compress the imprecision of the analysis below some given error bound.

Interestingly, the program transformation that provides a fix-point characterization of the covering for the class of locally complete programs in Theorem \ref{covering} has a clear correspondence with the {\em Relax\/} rule of LCL \cite{BruniGGR21}, a logic guaranteeing local completeness. The {\em Relax\/} rule is a key rule in LCL to ensure that pre/post conditions remain inside specific bounds in order to prove local completeness. This rule constrains the under-approximating post-condition to have the same abstraction as the strongest post-condition, which is needed for preserving local completeness, in a similar way as played by the condition $x\in \auco{}{A}(W_{i})~\wedge~\pfunction{e}{}{}(x){\downarrow}~\wedge~\pfunction{e}{}{}(x)\in \auco{}{A}(W_{h(e,i)})$ in (\ref{covering-function}).
An interesting open question is then whether LCL is sound and complete for the {\re} covering $\cC_i$ of $\bbbc(\cAA,i)$.
Other open questions remain such as the role of $\bbbc(\cAA,i)$ and $\compl{\bbbc(\cAA,i)}$ respectively in $\ok{\Pi^0_2}$ and $\ok{\Sigma^0_2}$. In this case we believe that $\bbbc(\cAA,i)$ is not $\ok{\Pi^0_2}$-hard.


\bibliographystyle{splncs04}

\end{document}